\documentclass[12pt,onecolumn,draftclsnofoot,journal]{IEEEtran}
\usepackage{graphicx}
\usepackage{float}
\usepackage{epstopdf}
\usepackage[cmex10]{amsmath}
\usepackage{array}
\usepackage{cite}
\usepackage{amssymb}
\usepackage{amsfonts}
\usepackage{amsmath}
\usepackage{stackrel}
\usepackage{booktabs,multirow}
\usepackage{arydshln}
\usepackage{slashbox}
\usepackage{amsthm}
\usepackage{listings}
\usepackage{algorithm}
\usepackage{algorithmicx}
\usepackage{algpseudocode}
\usepackage{tablefootnote}
\usepackage{threeparttable}
\usepackage{multicol}

\newtheorem{lem}{Lemma}
\newtheorem{rem}{Remark}
\newtheorem{theo}{Theorem}

\newtheorem{cor}{Corollary}
\newtheorem{pro}{Proposition}
\newtheorem{ex}{Example}
\newfloat{routine}{htbp}{loa}
\floatname{routine}{Routine} 

\makeatletter
\newcommand{\algmargin}{\the\ALG@thistlm}
\makeatother
\newlength{\forwidth}
\algdef{SE}[parFOR]{parFor}{EndparFor}[1]
  {\parbox[t]{\dimexpr\linewidth-\algmargin}{
     \hangindent\forwidth\strut\algorithmicfor\ #1\ \algorithmicdo\strut}}{\algorithmicend\ \algorithmicfor}
\algnewcommand{\parState}[1]{\State
  \parbox[t]{\dimexpr\linewidth-\algmargin}{\strut #1\strut}}

\newlength{\ifwidth}
\algdef{SE}[parIF]{parIf}{EndparIf}[1]
  {\parbox[t]{\dimexpr\linewidth-\algmargin}{
     \hangindent\ifwidth\strut\algorithmicif\ #1\ \algorithmicdo\strut}}{\algorithmicend\ \algorithmicif}

\begin{document}

\title{Characterization and Efficient Search of Non-Elementary Trapping Sets of LDPC Codes with Applications to Stopping Sets}
\author{Yoones Hashemi, and Amir H. Banihashemi,\IEEEmembership{ Senior Member, IEEE}
}

\maketitle


\begin{abstract}
In this paper, we propose a characterization for non-elementary trapping sets (NETSs) of low-density parity-check (LDPC) codes. The characterization is based on viewing a NETS as a hierarchy of embedded graphs starting from an ETS. The characterization corresponds to an efficient search algorithm that under certain conditions is exhaustive.  
As an application of the proposed characterization/search, we obtain lower and upper bounds on the stopping distance $s_{min}$ of LDPC codes. 
We examine a large number of regular and irregular LDPC codes, and demonstrate the efficiency and versatility of our technique in finding lower and upper bounds on, and in many cases the exact value of, $s_{min}$. 
Finding $s_{min}$, or establishing search-based lower or upper bounds, for many of the examined codes are out of the reach of any existing algorithm.
\end{abstract}

\section{introduction}

Finite-length LDPC codes under iterative decoding algorithms suffer from the \textit{error floor} phenomenon.  
It is well-known that the error-floor performance of LDPC codes is related to the presence of certain problematic graphical structures in the Tanner graph of the code, commonly referred to as \textit{trapping sets (TS)} \cite{richardson}.
Empirical results demonstrate that over the binary symmetric channel (BSC) and the additive white Gaussian noise channel (AWGNC),  the majority of error-prone structures are elementary TSs (ETS)~\cite{mehdi2014}, \cite{hashemireg}, \cite{hashemiireg}. These are TSs whose induced subgraphs contain only
degree-1 and degree-2 check nodes. In particular, the \textit{leafless ETSs (LETSs)}, in which each variable node is connected to at least two even-degree (satisfied) check nodes, are recognized as the main culprit in the error-floor of variable-regular LDPC codes~ \cite{hashemireg}. Most recently, in~\cite{hashemilower}, for variable-regular LDPC codes, lower bounds on the size of the smallest ETSs and TSs that are non-elementary (NETS) were established. It was shown in~\cite{hashemilower} that NETSs are generally larger than ETSs with the same number of odd-degree (unsatisfied) check nodes. This provided a theoretical justification, though not quite conclusive, for why ETSs often happen to be more harmful than NETSs. From a practical viewpoint, the ``elementary'' property simplifies the analysis and search of ETSs compared to NETSs, see,~\cite{mehdi2014}, \cite{hashemireg}, \cite{hashemiireg}, and the references therein. In particular, ETSs lend themselves to an alternate graphical representation, dubbed {\em normal graph}~\cite{mehdi2014}, that is simpler than the commonly used bipartite graph representation. Normal graphs have been used to develop the most efficient search algorithms for ETSs~\cite{mehdi2014}, \cite{hashemireg}, \cite{hashemiireg}. 

While empirical results have shown that the majority of harmful TSs over BSC and AWGNC are elementary, there are still some smaller NETSs that can trap iterative decoders over these channels in the error floor region.  
To the best of our knowledge, the \textit{branch-\&-bound} algorithm of \cite{zhang} is the only exhaustive search algorithm in the literature capable of finding both ETS and NETSs of LDPC codes.  The branch-\&-bound technique is a systematic enumeration of all candidate solutions that is commonly used to solve NP-hard integer programming problems.
Being a branch-\&-bound algorithm, 
the algorithm of \cite{zhang} is thus only capable of finding relatively small TSs with relatively small number of unsatisfied check nodes, and is only applicable to codes with short block lengths (the block lengths of all the reported codes are less than 1008). In this work, we propose an efficient search algorithm for NETSs of LDPC codes that has a much wider reach than branch-\&-bound-type algorithms in terms of both the code's block length and the size of TSs. The proposed search algorithm is graph-based, and relies on the characterization of a NETS as an embedded sequence of graphs that starts from an ETS, and expands one variable node at a time to reach the NETS.  The relatively low computational complexity of the proposed algorithm is a result of the simplicity of the expansions, and the fact that efficient algorithms already exist for finding ETSs~\cite{hashemireg}, \cite{hashemiireg}. One of the main contributions of this paper is to determine theoretically the range in which the proposed algorithm finds an exhaustive list of NETSs. 
%
%

As an important application of the proposed characterization/search of NETSs, we derive lower and upper bounds on the {\em stopping distance},  $s_{\min}$, of LDPC codes. {\em Stopping sets} (SS) are known to be the error-prone structures of LDPC codes over the binary erasure channel (BEC) under the belief propagation algorithm~\cite{Di}, and stopping distance is the size of the smallest stopping set(s).
%
It is well-known that, in general, finding $s_{\min}$ of an arbitrary LDPC code is an NP-hard problem \cite{krish}. 
Nevertheless, much research has been devoted to estimating/finding $s_{\min}$, and to obtaining a list of small stopping sets, for LDPC codes~\cite{richter}, \cite{Hu2}, \cite{Hirotomo2}, \cite{wang2009finding}, \cite{Rosnes}, \cite{Rosnes2}, \cite{Rosnes3}, \cite{esma}, \cite{But}. These results are mostly limited to codes with short to moderate block lengths and/or low to moderate rates and/or small variable degrees.
In \cite{wang2009finding}, the authors proposed a branch-\&-bound search algorithm to find small stopping sets. 
The proposed algorithm, however, becomes quickly infeasible to use as the block length, $n$, and $s_{\min}$ are increased. All the three LDPC codes studied in \cite{wang2009finding} are structured regular codes with variable degree $3$ and rate less than or equal to $0.5$. Using the Stern's probabilistic algorithm \cite{stern}, the authors in \cite{Hu2}, \cite{Hirotomo2} proposed search algorithms for computing $s_{\min}$ of LDPC codes. Their search algorithms, however, are also applicable only to short block length random codes or medium block length structured codes. Moreover, all the LDPC codes studied in \cite{Hu2}, \cite{Hirotomo2} have variable degree $3$ and rate less than or equal to $0.5$.
Authors in \cite{Rosnes} and \cite{Rosnes2} proposed  branch-\&-bound algorithms to find the stopping sets of LDPC codes.  To the best of our knowledge, this is the most efficient exhaustive search of stopping sets available in the literature. Similar to all the other branch-\&-bound algorithms, however, the computational complexity of these algorithms increases very rapidly with block length and thus the approach is only limited to short block lengths.
In particular, except for two random regular codes with rate $0.5$ and block length of $504$, all the codes studied in  \cite{Rosnes} and \cite{Rosnes2} are structured codes\footnote{The structural properties of the codes were used in \cite{Rosnes} and \cite{Rosnes2} to  speed up the search.} with length less than or equal to $4896$. 
Also, all the regular codes studied in  \cite{Rosnes} and \cite{Rosnes2} have variable degree $3$ and rate $0.5$.

In this paper, we use the graphical structure of stopping sets within the Tanner graph of an LDPC code to devise our search algorithm, and to derive bounds on $s_{\min}$. The subgraph induced by a stopping set in the Tanner graph of an LDPC code contains only check nodes with degree two or larger. We consider two categories of stopping sets depending on the check node degrees in their subgraph. If all the check nodes have degree two, we call the stopping set {\em elementary (ESS)}. Otherwise, the stopping set is referred to as {\em non-elementary (NESS)}.
Considering that an ESS is a LETS with no unsatisfied check node, 
we use the highly efficient algorithms of~\cite{hashemireg},~\cite{hashemiireg}, to search for ESSs.  NESSs, on the other hand, are a subset of NETSs. To search for NESSs, we thus use the proposed search algorithm for NETSs. 
Despite the fact that the proposed algorithms here are highly efficient, the exhaustive search of stopping sets of large size for longer LDPC codes may still happen to be too complex to perform. For a manageable complexity, we thus derive  a bound $L$ on the size of stopping sets that can be searched exhaustively. If the exhaustive search within this range results in finding at least one stopping set, then the smallest size of such stopping sets is $s_{\min}$. Otherwise, we establish the lower bound of $L$ on  $s_{\min}$. In this case, we modify our search algorithms to further reduce their complexity but at the expense of sacrificing the exhaustiveness. We then use the modified algorithms to search for stopping sets of size larger than $L$. The smallest size of such stopping sets is used as an upper bound on  $s_{\min}$.
In general, if we succeed in finding the exact stopping distance of an LDPC code, we do so in much higher speed than existing algorithms. If we fail, and establish the lower bound of $L \leq s_{\min}$, 
our algorithm for finding an upper bound is often much faster than the existing algorithms, for example, those in \cite{richter}, \cite{Hu2}, \cite{Hirotomo2}, \cite{wang2009finding}. We provide extensive numerical results that demonstrate the application of our technique to a variety of regular and irregular LDPC codes with block lengths as large as more than $16,000$. In fact, one of the main advantages of our search algorithms is that, unlike 
the existing algorithms in the literature such as \cite{Hirotomo2}, \cite{Hu2}, \cite{Rosnes}, the complexity does not change much by increasing the block length.

The rest of the paper is organized as follows. Basic definitions and notations are provided in Section \ref{sec:pre}.  We also briefly explain the search algorithms of~\cite{hashemireg} and~\cite{hashemiireg} for finding LETSs/ETSs in this section.
Section II ends with revisiting the lower bounds derived in  \cite{hashemilower} on the smallest size of ETSs and NETSs. In Section \ref{sec:searTS}, we present the characterization of NETS structures and propose an efficient exhaustive/non-exhaustive search of NETSs for regular and irregular LDPC codes.  In Section \ref{sec:bounds}, 
we discuss the derivation of lower and upper bounds on the stopping distance of LDPC codes.
Finally, numerical results are provided in Section \ref{sec:num}, followed by concluding remarks in Section \ref{sec:conclude}.

\section{Preliminaries}
\label{sec:pre}

\subsection{Definitions and Notations}

Consider an undirected graph $G=(F, E)$, where the two sets $F=\{f_1,\dots,f_k\}$ and $E=\{e_1,\dots,e_m\}$, are the sets of \textit{nodes} and \textit{edges} of $G$,
respectively. We say that an edge $e$ is \textit{incident} to a node $f$ if $e$ is connected to $f$.  If there exists an edge $e_k$ which is incident to 
two distinct nodes $f_i$ and $f_j$, we represent $e_k$ by $f_i f_j$ or $f_j f_i$.   
The degree of a node $f$ is denoted by $d_f$, and is defined as the number of edges incident to $f$. 
The \textit{minimum degree} of a graph $G$,  denoted by $\delta (G)$ is defined to be the  minimum 
degree of its nodes. A node $f$ is called \textit{leaf} if $d_f=1$.  A \textit{leafless graph} $G$ is a graph with $\delta (G) \geq 2$.

Given an undirected graph $G=(F,E)$,  a \textit{walk} between two nodes $f_1$ and $f_{k+1}$ is a sequence of nodes and edges
$f_1$, $e_1$, $f_2$, $e_2$, $\dots$, $f_k$, $e_k$, $f_{k+1}$, where $e_i=f_i f_{i+1}$, $\forall i \in [1,k]$. 
A {\em path} is a walk  with no repeated nodes or edges, except the first and the last nodes that can be the same. 
If the first and the last nodes are distinct, we call the path an {\em open path}.
Otherwise, we call the path a {\em cycle}. The \textit{length} of a walk, a  path, or a cycle is the number of its edges. 
A {\em lollipop walk} is a walk $f_1$, $e_1$, $f_2$, $e_2$, $\dots$, $f_k$, $e_k$, $f_{k+1}$, such that all the edges and all the nodes are distinct, except that $f_{k+1}=f_m$, for some $m \in (1,k)$.
A \textit{chord} of a cycle is an edge which is not part of the cycle but is incident to two distinct nodes in the cycle.
A {\em chordless} or \textit{simple cycle} is a cycle which does not have any chord. 
 The length of the shortest cycle(s) in a graph is called {\em girth}, and is denoted by $g$.
A graph is called {\em connected} when there is a {\em path} between every pair of nodes in the graph. 
A \textit{tree} is a connected graph that contains no cycles. 
A \textit{rooted tree} is a tree in which one specific node is assigned as the \textit{root}. 
The \textit{depth of a node} in a rooted tree is the length of the
path from the node to the root. The \textit{depth of a tree} is the maximum depth of any node in the tree. 
\textit{Depth-one tree (dot)} is a tree with depth one.

Any $m \times n$ parity check matrix $H$ of a binary LDPC code $\mathcal{C}$ can be represented by its bipartite Tanner graph $G=(V \cup C, E)$, 
where $V=\{ v_1,v_2,\dots,v_n \}$ is the set of variable nodes and $C=\{ c_1,c_2,\dots,c_m \}$ is the set of check nodes. If there is an edge $e \in E$ between the nodes $v_i$ and $c_j$ in the Tanner graph, 
then, correspondingly, there is a ``1'' in the $(j,i)$-th entry of matrix $H$.
A Tanner graph is called {\em variable-regular} with variable degree $d_{\mathrm{v}}$ if $d_{v_i} = d_{\mathrm{v}}$, $\forall~{v}_{i} \in V$. 
A Tanner graph is called {\em irregular} if it has multiple variable and check node degrees. 
An irregular LDPC code is usually described by its variable and check node degree distributions, $\lambda(x)=\sum\limits_{i=d_{v_{\min}}}^{d_{v_{\max}}} \lambda_i x^{i-1}$ and $\rho(x)=\sum\limits_{i=d_{c_{\min}}}^{d_{c_{\max}}} \rho_i x^{i-1}$, respectively, where $\lambda_i$ and $\rho_i$ are the fractions of edges in the Tanner graph that are incident to degree-$i$ variable and check nodes, respectively. The terms $d_{v_{\max}}, d_{v_{\min}}$ ($d_{c_{\max}}, d_{c_{\min}}$) 
are the maximum and minimum degrees of variable nodes (check nodes), respectively. For variable-regular Tanner graphs, we have $ d_{\mathrm{v}} \geq 2$, and for irregular 
ones, we assume $d_{v_{\min}} \geq 2$. The girth of a Tanner graph is an even number and in this work, we study Tanner graphs that are free of 4-cycles ($g > 4$). 

For a subset $\mathcal{S}$ of $V$, the subset $\Gamma{(\mathcal{S})}$ of $C$ denotes the set of neighbors of $\mathcal{S}$ in $G$.
The \textit{induced subgraph} of $\mathcal{S}$ in $G$, denoted by $G(\mathcal{S})$, is the graph with the set of nodes $\mathcal{S} \cup \Gamma{(\mathcal{S})}$ and 
the set of edges $\{f_i f_j \in E : f_i \in \mathcal{S}, f_j \in \Gamma{(\mathcal{S})}\}$. The set of check nodes with odd and even degrees in $G(\mathcal{S})$ are denoted by $\Gamma_{o}{(\mathcal{S})}$ and $\Gamma_{e}{(\mathcal{S})}$, respectively. In this paper, the terms \textit{unsatisfied check nodes} and \textit{satisfied check nodes} are used to refer to the check nodes in
$\Gamma_{o}{(\mathcal{S})}$ and $\Gamma_{e}{(\mathcal{S})}$, respectively. 
The \textit{size} of an induced subgraph $G(\mathcal{S})$ is defined to be the number of its variable nodes.  
We assume that an induced subgraph is connected. Disconnected subgraphs can be considered as the union of connected ones. 


Given a Tanner graph G, a set $\mathcal{S}\subset V$ is called an \textit{(a,b) trapping set (TS)} if $|\mathcal{S}| = a$ and $|\Gamma_{o}{(\mathcal{S})}| = b$.
Alternatively, $\mathcal{S}$ is said to belong to the {\em class of (a,b) TSs}. Parameter $a$ is referred to as the {\em size} of the TS.
 An \textit{elementary trapping set (ETS)} is a trapping set for which all the check nodes in $G(\mathcal{S})$ have degree 1 or 2. 
To simplify the representation of ETSs, similar to ~\cite{mehdi2014}, \cite{yoones2015}, \cite{hashemireg}, we use an alternate graph representation of ETSs, called \textit{normal graph} in variable-regular graphs. The normal graph of an ETS $\mathcal{S}$ is obtained from $G(\mathcal{S})$ by removing all the check nodes of degree one and their incident edges, and by replacing all the  degree-2 check nodes and their two incident edges by a single edge. 
 We call a set $\mathcal{S}\subset V$ an \textit{(a,b) leafless ETS (LETS)} if $\mathcal{S}$ is an $(a,b)$
 ETS and if the normal graph of $\mathcal{S}$ is leafless.
 Otherwise, the set ${\cal S}$ is called an ETS with leaf (ETSL). 
A \textit{non-elementary trapping set (NETS)} is a trapping set which is not elementary. A \textit{stopping set (SS)} is a trapping set for which  $G(\mathcal{S})$ has no check node of degree one. In general, similar to the TSs, SSs can be partitioned into two categories of \textit{elementary SSs (ESSs)} and \textit{non-elementary SSs (NESSs)}.
 
The following lemma shows that for variable-regular LDPC codes, depending on $d_{\mathrm{v}}$ being odd or even, some classes of trapping sets cannot exist.

\begin{lem}
\label{rem:cannot}
\cite{hashemilower} In a variable-regular Tanner graph with variable degree $d_{\mathrm{v}}$, (a) if $d_{\mathrm{v}}$ is odd, then there does not exist any $(a,b)$ TS with odd $a$ and even $b$, or with even $a$ and odd $b$;
and (b) if $d_{\mathrm{v}}$ is even, then there does not exist any $(a,b)$ TS with odd $b$.
\end{lem}

%
%

\subsection{Exhaustive Search of ETSs}
\label{sec:searETS}
In  \cite{hashemireg},  a hierarchical graph-based expansion approach was proposed to characterize LETSs of variable-regular LDPC codes. It was proved in~\cite{hashemireg} that any LETS structure of variable-regular Tanner graphs for any variable degree $d_{\mathrm{v}}$, and in any $(a,b)$ class, can be generated by applying a combination of \textit{depth-one tree (dot),  path} and \textit{lollipop} expansions to simple cycles.  Figs. \ref{fig:pathen}  (a)-(c) show the three expansions in the space of normal graphs. (Notations $pa_m^o$ and $pa_m^c$ are used for open and closed paths of length $m+1$, respectively. The notation $lo^c_m$ is used for a lollipop walk of length $m+1$ that consists of a cycle of length $c$.)
\begin{figure}[] 
\centering
\includegraphics [width=0.4\textwidth]{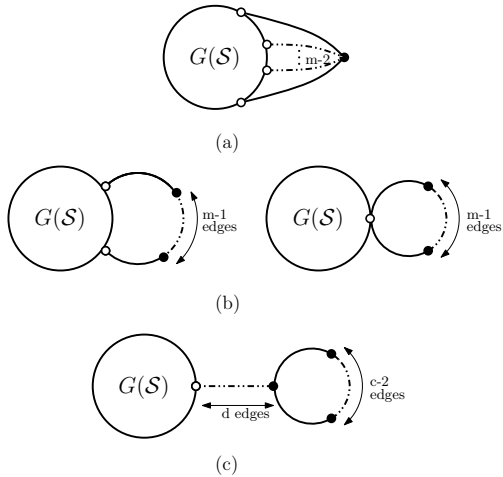}
\caption{Expansion of the LETS structure $\mathcal{S}$ with (a) a depth-one tree with $m$ edges, $dot_m$, (b) an open and closed $path$ of length $m+1, pa_m^o, pa_m^c$, respectively, (c) a lollipop walk of length $m+1=d+c$, $lo^c_m$.}
\label{fig:pathen}
\end{figure}
The characterization, dubbed as $dpl$, was then used as a road map to devise search algorithms that are provably efficient in finding all the instances of $(a,b)$ LETS structures with $a \leq a_{\max}$ and $b \leq b_{\max}$, for any choice of $a_{\max}$ and $b_{\max}$, in a guaranteed fashion. The $dpl$ search algorithm starts by enumerating short simple cycles in the graph and then searches for the children (descendants) of those cycles through the three expansion techniques recursively, until it reaches the targeted structure.

In \cite{hashemiireg}, LETSs and ETSLs were studied in irregular Tanner graphs. It was shown that these structures in irregular graphs can also be characterized and searched using a $dpl$-based technique in any interest range, efficiently and exhaustively.
In the $dpl$ characterization/search algorithm of~\cite{hashemireg} and \cite{hashemiireg}, to exhaustively cover all the $(a,b)$ LETS structures in the interest range of  $a \leq a_{\max}$ and $b \leq b_{\max}$, the algorithm sometimes 
 needs to also cover auxiliary structures with their $b$ values larger than $b_{\max}$, and up to $b'_{\max} > b_{\max}$. 
 
 To the best of our knowledge, the $dpl$-based algorithms of \cite{hashemireg} and \cite{hashemiireg} are the most efficient exhaustive search algorithms available for finding ETSs of LDPC codes.

\subsection{Lower bounds on the size of TSs}
 
\begin{theo}
\label{lowf}
 \cite{hashemilower} Consider a variable-regular Tanner graph $G$ with variable degree $d_\mathrm{v}$ and girth $g$. 
A lower bound on the size $a$ of an $(a,b)$ trapping set in $G$, whose induced subgraph contains a check node of degree $k\:(\geq 2)$ is given in (\ref{ineq1}), 
where $b'=b-(k \mod 2)$, $T=k(d_\mathrm{v}-1)-b'$, and $b$ is assumed to satisfy $b < k(d_\mathrm{v}-1)+(k \mod 2)$. (Notation $\mod$ is for modulo operation.)
\begin{table*}
\begin{align}
a \geq 
\left\{
\begin{array}{ll}
k+ T \sum\limits_{i=0}^{\lfloor g/4 -2  \rfloor}  (d_\mathrm{v}-1)^i,    & \emph{for}~g/2~\emph{even}, \\
k+T \sum\limits_{i=0}^{\lfloor g/4 -2  \rfloor} (d_\mathrm{v}-1)^i+max \{\lceil (T (d_\mathrm{v}-1)^{\lfloor g/4-1\rfloor})/d_\mathrm{v} \rceil, (d_\mathrm{v}-1-\lfloor b'/ k \rfloor)(d_\mathrm{v}-1)^{\lfloor g/4-1\rfloor}\},   & \emph{for}~g/2~ \emph{odd},
\end{array}
\right.
\label{ineq1}
\end{align}
\end{table*}
\end{theo}

The proof of the above result, presented in~\cite{hashemilower}, is based on considering the tree-like expansion of the induced subgraph of the TS starting from the degree-$k$ check node $w$ as the root with $g/2$ layers (depth $g/2-1$).  
Fig. \ref{treegen} shows such an expansion. (Variable nodes, satisfied and unsatisfied check nodes are represented by circles, empty and full squares, respectively.)
\begin{figure}[] 
\centering
\includegraphics [width=0.45\textwidth]{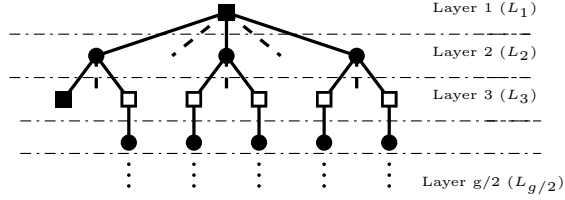}
\caption{A tree-like expansion of a TS rooted at a check node of degree $k$.}
\label{treegen}
\end{figure}
In this tree, node $w$, in layer one ($L_1$), is connected to $k$ variable nodes in layer two ($L_2$), and each variable node in $L_2$ is connected to $d_\mathrm{v}-1$ check nodes in $L_3$. Also, 
to minimize the size of TSs, it is assumed that the degree of all the other check nodes in the subgraph is either $2$ or $1$. From Theorem \ref{lowf}, one can see that for any given values of $b$, $g$ and $d_\mathrm{v}$, 
by increasing $k$, the lower bound on the size of the smallest TS is increased. 

\begin{rem}
\label{cor:fargh}
Based on Lemma \ref{rem:cannot}, for odd values of $d_\mathrm{v}$, there is no TS in the $(a,0)$ class with an odd value of $a$. In such cases, therefore, the lower bound of 
Theorem~\ref{lowf} can be improved, if the value in (\ref{ineq1}) happens to be an odd number.
\end{rem}

\begin{cor}
 \cite{hashemilower}
 \label{cor:nets}
Theorem \ref{lowf} with $k=3$ ($k=4$) provides a lower bound on the size of the smallest possible NETSs with $b > 0$ ($b=0$).
\end{cor}

It was shown in \cite{hashemilower} that the lower bounds of Corollary \ref{cor:nets} are often tight. The size of the smallest ETSs was also compared in \cite{hashemilower} with the lower bounds of 
Corollary \ref{cor:nets}. The following result follows from Table I of \cite{hashemilower}.

\begin{rem}
\label{cor:LET,NET}
For any given $b \leq 5$, $d_\mathrm{v}=3,4,5,6$, and $g=6,8,10$, the smallest possible TSs with cycles are LETSs.
\end{rem}

The following lemma is easy to prove based on an approach similar to the one used to prove Theorem~\ref{lowf}.

\begin{lem}
\label{lem:comb1}
The lower bound of Theorem~\ref{lowf} also applies to the size $a$ of an $(a,b)$ NETS whose induced subgraph contains \underline{at least} one check node of degree $k\:(\geq 3)$.
\end{lem}



\section{Characterization/Search of Non-Elementary TSs (NETSs) in LDPC Codes}
\label{sec:searTS}

\subsection{Characterization and Exhaustive Search of NETSs in Variable-Regular LDPC Codes}
\label{sec:nets}
The characterization of ETSs (LETSs and ETSLs) for variable-regular graphs, provided in \cite{hashemireg} and \cite{hashemiireg}, is based on normal graph representation of structures. This approach, however, is not applicable to NETSs. In this work, to develop the characterization of NETSs, we investigate the parent-child relationships 
between ETSs and NETSs. As natural candidates for the expansion of ETSs to reach NETSs, we consider $dot$, $path$ and $lollipop$ expansions. One can see that the application of $path$ and $lollipop$ expansions to a TS increases the $b$ value of the structure rather rapidly. For NETSs with relatively small $b$ values, we thus limit the expansions to $dot$ in the rest of the paper. Due to the low computational complexity of $dot$ expansion \cite{hashemireg}, this results in an efficient NETS search algorithm starting from ETSs. Using only the $dot$ expansion limits the variety of NETS structures that can be generated starting from ETS structures. In the following, we first discuss the (successive) application(s) of $dot$ expansions to ETS structures and then describe the NETS structures that are out of reach.

Suppose that $\mathcal{S}$ is an $(a,b)$ TS structure of variable-regular Tanner graphs with variable degree $d_\mathrm{v}$, where $b \geq 1$. The notation $dot_m$ is used for a {\em dot} expansion with $m$ edges, connecting a new variable node to $m$ check nodes of $\mathcal{S}$.
Similar to \cite{hashemireg}, we assume that the new variable node in $dot_m$ is connected to at least two check nodes of $\mathcal{S}$, i.e., $2 \leq m \leq d_\mathrm{v}$. 
However, unlike the $dot_m$ used in \cite{hashemireg}, the $m$ edges can be connected to both satisfied and unsatisfied check nodes of $\mathcal{S}$. The following result is simple to prove.

\begin{lem}
\label{pro:dot}
Suppose that $\mathcal{S}$ is an $(a,b)$ TS structure of variable-regular Tanner graphs with variable degree $d_\mathrm{v}$, where $b \geq 1$.
Expansion of $\mathcal{S}$ using $dot_m$, $ 2 \leq m \leq d_\mathrm{v}$, will result in 
NETS structure(s) in the  $(a+1,b+d_\mathrm{v}-2q)$ class, where $m=p+q$ and $p \geq 0$ and $q \geq 0$ are the number of edges connecting the new variable node to the satisfied and unsatisfied check nodes of $\mathcal{S}$, respectively. 
\end{lem} 
 
 \begin{rem}
 One should note that in Lemma \ref{pro:dot}, if  $\mathcal{S}$ is an ETS, then $p>0$.
 \end{rem}
 
 \begin{lem}
 \label{lem3}
 Consider the induced subgraph $G(\mathcal{S})$ of a NETS $\mathcal{S}$ in a variable-regular Tanner graph $G$ with variable degree $d_\mathrm{v}$. Further, consider the expansions of  this subgraph in a layered tree-like fashion 
 starting from one of the check nodes $w$ with degree $k$, where $k \geq 3$. If any such expansion can be partitioned into $k$ subgraphs, where the only connection of the subgraphs is through $w$, as shown in Fig.~\ref{discon}, then the NETS structure ${\cal S}$ cannot be generated by (successive) application(s) of $dot_m$ expansions, $2 \leq m \leq d_\mathrm{v}$, to any ETS structure.  
 \end{lem}
 
 \begin{proof}
Consider a NETS structure ${\cal S}$ whose subgraph satisfies the condition of the lemma, i.e., there is a tree-like expansion of  $G(\mathcal{S})$ rooted at a degree-$k$ check node  ($k \geq 3$) with $k$ disconnected subgraphs $\mathcal{S}_1,...,\mathcal{S}_k$, as shown in Fig. \ref{discon}. In Fig. \ref{discon}, the degree-$k$ check node at the root (first layer $L_1$) is connected to $k$ variable nodes at layer $L_2$. Those variable nodes are each connected to 
$d_\mathrm{v} - 1$ other check nodes at $L_3$ and so on. It is straightforward to see that a structure with the subgraph of Fig. \ref{discon} cannot be generated through successive applications of $dot_m, m \geq 2$, to an ETS structure ${\cal S}'$. The reason is that to create the check node $w$ (with degree $k \geq 3$) in the process of expansion, there are two possibilities: ($i$) Node $w$ belongs to ${\cal S}'$, or ($ii$) it is added in the expansion process. In Case ($i$), the degree of $w$ in ${\cal S}'$ is either one or two. For the degree of $w$ to be increased to $k$ in the expansion process, through one or more $dot_m$ expansions, one or more variable nodes will have to be added to the subgraph, each with one connection to $w$ and with one or more connection(s) to the other check nodes of the existing (connected) subgraph.  This is in contradiction with the structure in Fig~\ref{discon}, where otherwise disconnected subgraphs $\mathcal{S}_1,...,\mathcal{S}_k$ are only connected through $w$. The proof for Case ($ii$) is similar.
\begin{figure}[] 
\centering
\includegraphics [width=0.45\textwidth]{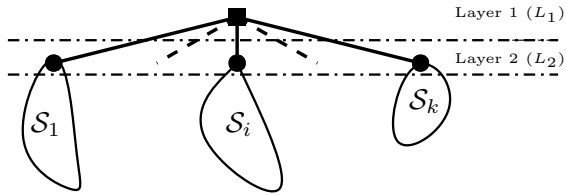}
\caption{A NETS with $k$ disconnected subgraphs.}
\label{discon}
\end{figure}
 \end{proof}

In the following lemma, we investigate the smallest size of NETS structures with induced subgraphs of the form discussed in Lemma~\ref{lem3} and presented in Fig.~\ref{discon}, for different values of $d_\mathrm{v}$, $g$, and $b \leq 4$.
Similar results may be derived for other variable degrees, girths and $b$ values. 

\begin{lem}
\label{lem4}
For variable-regular graphs with $d_\mathrm{v}=3,4,5,6$, $g=6,8,10$, and for $b \leq 4$, the size $a$ of the smallest possible NETS in the $(a,b)$ class containing disconnected subgraphs (as shown in Fig. \ref{discon}) is listed in Table \ref{tab:disc}. 
\end{lem}

\begin{table*}[]
\centering
\renewcommand{\arraystretch}{1.1}
\setlength{\tabcolsep}{2pt}
\caption{The smallest size $a$ of NETSs in the $(a,b)$ class ($b \leq 4$) with disconnected subgraphs for Tanner graphs with $d_\mathrm{v}=3,4,5,6$, $g=6,8,10$}
\label{tab:disc}
\begin{tabular}{||c|c|c|c||c|c|c||c|c|c||c|c|c||}
\cline{1-13}
&\multicolumn{3}{c||}{$d_\mathrm{v}=3$}&\multicolumn{3}{c||}{$d_\mathrm{v}=4$}&\multicolumn{3}{c||}{$d_\mathrm{v}=5$}&\multicolumn{3}{c||}{$d_\mathrm{v}=6$}\\
\cline{2-13}
&$g=6$&$g=8$&$g=10$&$g=6$&$g=8$&$g=10$&$g=6$&$g=8$&$g=10$&$g=6$&$g=8$&$g=10$\\
\cline{1-13}
$b=1$&$\begin{array}{@{}c@{}}15 \end{array} $&$\begin{array}{@{}c@{}}21 \end{array} $&$\begin{array}{@{}c@{}}27 \end{array} $&\multicolumn{3}{c||}{not possible}&$\begin{array}{@{}c@{}}21  \end{array} $&$\begin{array}{@{}c@{}}39 \end{array} $&$\begin{array}{@{}c@{}}93  \end{array} $&\multicolumn{3}{c||}{not possible}\\
\cline{1-13}
$b=2$&$\begin{array}{@{}c@{}} 14 \end{array} $&$\begin{array}{@{}c@{}} 20 \end{array} $&$\begin{array}{@{}c@{}} 26 \end{array} $&\multicolumn{3}{c||}{not possible}&$\begin{array}{@{}c@{}}20  \end{array} $&$\begin{array}{@{}c@{}}38  \end{array} $&$\begin{array}{@{}c@{}}92  \end{array} $&\multicolumn{3}{c||}{not possible}\\
\cline{1-13}
$b=3$&$\begin{array}{@{}c@{}} 11 \end{array} $&$\begin{array}{@{}c@{}} 15 \end{array} $&$\begin{array}{@{}c@{}}19 \end{array} $&\multicolumn{3}{c||}{not possible}&$\begin{array}{@{}c@{}}19  \end{array} $&$\begin{array}{@{}c@{}}37  \end{array} $&$\begin{array}{@{}c@{}}91  \end{array} $&\multicolumn{3}{c||}{not possible}\\
\cline{1-13}
$b=4$&$\begin{array}{@{}c@{}} 10 \end{array} $&$\begin{array}{@{}c@{}} 14 \end{array} $&$\begin{array}{@{}c@{}} 18 \end{array} $&$\begin{array}{@{}c@{}}15  \end{array} $&$\begin{array}{@{}c@{}} 24  \end{array} $&$\begin{array}{@{}c@{}}57  \end{array} $&$\begin{array}{@{}c@{}}18  \end{array} $&$\begin{array}{@{}c@{}}36  \end{array} $&$\begin{array}{@{}c@{}}90  \end{array} $&$\begin{array}{@{}c@{}}21  \end{array} $&$\begin{array}{@{}c@{}}36 \end{array} $&$\begin{array}{@{}c@{}}90  \end{array} $\\
\cline{1-13}
\end{tabular}
\end{table*} 

\begin{proof}

Suppose that ${\cal S}^*$ is the smallest NETS structure with disconnected subgraphs. Based on Corollary \ref{cor:nets}, the structure ${\cal S}^*$ contains a degree-$3$ check node at $L_1$. 
Consider each subgraph $\mathcal{S}_1$,  $\mathcal{S}_2$ and $\mathcal{S}_3$ as a TS containing the degree-3 check node at $L_1$ as a degree-$1$ (unsatisfied) check node (Fig. \ref{discon}). 
Let  $a_{\mathcal{S}_i}$ and $b_{\mathcal{S}_i}$ be the size and the number of unsatisfied check nodes of $\mathcal{S}_i$, respectively. 
Clearly, $b_{\mathcal{S}_i}>0$. We also have

\begin{align}
 b_{\mathcal{S}_1}+b_{\mathcal{S}_2}+b_{\mathcal{S}_3}= b+2, 
 \label{ineq2}
\end{align}

For ${\cal S}^*$ to have the smallest size, we look for $\mathcal{S}_i$s with smallest $a_{\mathcal{S}_i}$ whose $b_{\mathcal{S}_i}$ values satisfy $b_{\mathcal{S}_i}>0$, and the constraint (\ref{ineq2}).  
It is easy to see that the most favorable candidate for TSs $\mathcal{S}_i$s is an ETSL with no cycle. These structures  are denoted by $ETSL_2$ in~\cite{hashemiireg}, and exist only in the $(a, b= a(d_\mathrm{v}-2)+2)$ class.
If an $ETSL_2$ structure is not possible (due to the specific choice of $b_{\mathcal{S}_i}$), then based on Remark \ref{cor:LET,NET}, a LETS structure is the next favorable choice. 
To find the size of ${\cal S}^*$, therefore, one needs to consider all the possible combinations of positive integers $b_{\mathcal{S}_1}$, $b_{\mathcal{S}_2}$ and $b_{\mathcal{S}_3}$ that satisfy (\ref{ineq2}), 
and for each combination finds the smallest values of $a_{\mathcal{S}_i}$s using the aforementioned guidelines.
In the following, we prove the result for the case of  $d_\mathrm{v}=3$, $g=8$ and $b= 2,3$. The proof for the other cases listed in Table~\ref{tab:disc} is similar.

For $b=2$, using (\ref{ineq2}), we have $ b_{\mathcal{S}_1}+b_{\mathcal{S}_2}+b_{\mathcal{S}_3}=4$. The only positive integers satisfying this equality are $\{1,1,2\}$. Since, for $d_\mathrm{v}=3$, there does not exist any $ETSL_2$ with 
$b < 3$, then we look for LETS structures of minimum size with these $b$ values. From Table I in \cite{hashemilower}, the size of the smallest LETSs with  $b_{\mathcal{S}_i}=1$ and $b_{\mathcal{S}_i}=2$ is $7$ and $6$, respectively.
We thus conclude that the size of ${\cal S}^*$ is $7 + 7 + 6 = 20$.

For $b=3$, we have $b_{\mathcal{S}_1}+b_{\mathcal{S}_2}+b_{\mathcal{S}_3}=5$. The only solutions to this equation are $\{1,2,2\}$ and $\{1,1,3\}$. Again for the first set, no $ETSL_2$ structure exists, and based on LETS structures of minimum size, we obtain $7+6+6=19$ as the size of the corresponding NETS structure. For the second set of $b_{\mathcal{S}_i}$ values, we select an $ETSL_2$ structure for the $b_{\mathcal{S}_i}$ value $3$. This corresponds to $a_{\mathcal{S}_i}=1$.
For the other two TSs, the minimum size LETS structures have size $7$, and thus the size of the corresponding NETS structure in this case is $1+7+7=15$. Since $15$ is the smaller value between $19$ and $15$, it is in fact the size of ${\cal S}^*$. 

To obtain the entries in Table~\ref{tab:disc} for even values of $d_\mathrm{v}$, one should note that based on Lemma \ref{rem:cannot}, it is not possible to have a TS with even $d_\mathrm{v}$ and odd $b$. 
Therefore, for even values of $d_\mathrm{v}$, the minimum value of $b_{\mathcal{S}_i}$ is $2$, and in (\ref{ineq2}), the smallest value of $b$ for NETS structures under consideration is strictly larger than $3$.
\end{proof}

In the following, we investigate the parent-child relationships between ETSs and NETSs based on $dot_m$ expansions. Since the NETS structures discussed in Lemmas~\ref{lem3} and~\ref{lem4} are 
excluded, in the rest of the paper, we use the expression ``interest range of $a$ and $b$'' or ``$(a,b)$ class of interest'' to mean the $b$ values that satisfy $b \leq 4$, and for a given $b$ value in this range, the value of $a$ being strictly less than 
the entry provided in Table \ref{tab:disc}.
 
\begin{pro}
\label{lem:nets3}
Any  NETS structure $\mathcal{S}$ of variable-regular graphs with variable degree $d_\mathrm{v}$ in an $(a,b)$ class of interest, containing only one degree-$3$ check node (the rest of the check nodes have degree $2$ or $1$) 
can be characterized by the application of a $dot_m$ expansion ($2 \leq m \leq d_\mathrm{v}$) to an ETS substructure, $\mathcal{S}'$,  in the $(a-1,b+2m-2-d_\mathrm{v})$ class, where $m$ is number of edges connecting  
the variable node in $\mathcal{S} / \mathcal{S}'$  to one degree-$2$ and $m-1$ degree-$1$ check nodes of $\mathcal{S}'$.
\end{pro}

\begin{proof}
The structure $\mathcal{S}$ contains only one check node $w$ of degree $3$.  We consider the tree-like expansion of $\mathcal{S}$ from $w$ as the root at $L_1$.
Based on the knowledge that this expansion of $\mathcal{S}$ does not consist of disconnected subgraphs as shown in Fig.~\ref{discon}, there must exist two variable nodes, say $v_1$ and $v_2$ at $L_2$ that are connected through a path that does not pass through $w$. Now, consider removing one of these two variable nodes, say $v_1$, and all its incident edges from $\mathcal{S}$. The remaining graph $\mathcal{S}'$ is still connected and has no check node with degree larger than 2, i.e., $\mathcal{S}'$ is an ETS. It is easy to see that $\mathcal{S}$ can be obtained by expanding $\mathcal{S}'$ by $v_1$ through a $dot_m$ ($2 \leq m \leq d_\mathrm{v}$) expansion. 
The class of $\mathcal{S}'$ can be obtained by using Lemma \ref{pro:dot}, assuming $p=1$.
\end{proof}

The following corollary describes the exhaustive search of NETSs with only one degree-$3$ check node.
 
\begin{cor}
\label{cor:nets3}
In variable-regular Tanner graphs with variable degree $d_\mathrm{v}$, all the $(a,b)$ NETSs containing only one degree-$3$ check node in the interest range of $a \leq a_{\max}$ and $b \leq b_{\max}$ ($a_{\max}$ less than the value in Table \ref{tab:disc} and $b_{max} \leq 4$) can be found by applying $dot_m$  expansions to all the ETSs in the range of $a \leq a_{\max}-1$ and $b \leq b'_{\max}= b_{\max}+d_\mathrm{v}-2$.
\end{cor}

\begin{pro}
\label{lem:nets33}
Any NETS structure $\mathcal{S}$  in an interest class of $(a,b)$ for variable-regular graphs with variable degree $d_\mathrm{v}$, that contains two degree-$3$ check nodes (the rest are degree-2 or -1) 
can be characterized by a $dot_m$ expansion ($2 \leq m \leq d_\mathrm{v}$) applied to one of the two following substructures $\mathcal{S}'$:
(i)  an ETS in the $(a-1,b+2m-4-d_\mathrm{v})$ class,  where from $m$ edges connecting the variable node in $\mathcal{S} / \mathcal{S}'$  to $\mathcal{S}'$, two and $m-2$ are connected to degree-$2$ and degree-$1$ check nodes of $\mathcal{S}'$, respectively; or 
(ii) a NETS containing one degree-$3$ check node in the $(a-1,b+2m-2-d_\mathrm{v})$ class, where from $m$ edges connecting the variable node in $\mathcal{S} / \mathcal{S}'$  to $\mathcal{S}'$, one and $m-1$ are connected to degree-$2$ and degree-$1$ check nodes of $\mathcal{S}'$, respectively.
\end{pro}

\begin{proof}
Consider the expansion of the NETS structure starting from one of the degree-$3$ check nodes $w$ at $L_1$. In the expansion, the other
degree-$3$ check node is located either at $L_3$ or at $L_{2i+1}$, where $i>1$. With an argument similar to the one presented in the proof of Proposition \ref{lem:nets3}, there exist two variable nodes $v_1$ and $v_2$ 
at $L_2$ such that there is a path between them that does not pass through $w$. Therefore, by removing one of these two variable nodes, say, $v_1$, and all its incident edges from ${\cal S}$, the resulted subgraph ${\cal S}'$ remains connected. 
Now if the second degree-$3$ check node was at $L_3$ and connected to $v_1$, then there remains no check node with degree larger than $2$ after the removal of $v_1$, i.e., the subgraph $\mathcal{S}'$ is an ETS.
On the other hand, if the other degree-$3$ check node was at $L_3$ but not connected to $v_1$ or it was at $L_{2i+1}$ with $i>1$, then the resulted subgraph $\mathcal{S}'$ is a NETS containing one degree-$3$ check node. 
In either case, structure $\mathcal{S}$ can be obtained by applying a $dot_m$ expansion ($2 \leq m \leq d_\mathrm{v}$) to $\mathcal{S}'$. The class of $\mathcal{S}'$ can be determined in each case by using Lemma~\ref{pro:dot}, assuming $p=2$ and $p=1$, respectively.
\end{proof}

The following result is a generalization of Proposition~\ref{lem:nets33}.

\begin{pro}
\label{lem:netsf}
Any NETS structure $\mathcal{S}$  in an interest class of $(a,b)$ for variable-regular graphs with variable degree $d_\mathrm{v}$, that contains $f \geq 2$ degree-$3$ check nodes (the rest are degree-$2$ or -$1$) 
can be characterized by a $dot_m$ expansion ($2 \leq m \leq d_\mathrm{v}$) applied to one of the $\eta = \min\{m,f\}$ following substructures $\mathcal{S}'$: For any value of $p$ in the range $1 \leq p \leq \eta$, substructure
$\mathcal{S}'$ is in the $(a-1,b+2(m-p)-d_\mathrm{v})$ class,  where from $m$ edges connecting the variable node in $\mathcal{S} / \mathcal{S}'$  to $\mathcal{S}'$, $p$ and $m-p$ are connected to degree-$2$ and degree-$1$ check nodes of $\mathcal{S}'$, respectively.
\end{pro}

Based on the above results, it is easy to see that a NETS structure with $f$ degree-$3$ check nodes can be generated through successive applications of $dot_m$ expansions to ETS structures. 
For this to correspond to an exhaustive search of such NETSs, the following corollary, that generalizes Corollary~\ref{cor:nets3}, provides the range of ETSs that need to be included.
%
%

\begin{cor}
\label{cor:netsf}
In variable-regular Tanner graphs with variable degree $d_\mathrm{v}$, all the NETSs containing $f$ degree-$3$ check nodes (the rest are degree-$2$ or -$1$) in the interest range of $a \leq a_{\max}$ and $b \leq b_{\max}$ ($a_{\max}$ less than the value in Table \ref{tab:disc} and $b_{max} \leq 4$) can be found by $f$ successive applications of $dot_m$  expansions to all the ETSs in the range of $a \leq a_{\max}- \lceil f /d_v \rceil$ and $b \leq b'_{\max}= b_{\max}+f(d_\mathrm{v}-2)$.
\end{cor}

%

The following results can all be proved similar to the cases involving NETSs with only degree-$3$ check nodes. The proofs are thus omitted to avoid redundancy.

\begin{pro}
\label{lem:nets4}
Any NETS structure $\mathcal{S}$  in an interest class of $(a,b)$ for variable-regular graphs with variable degree $d_\mathrm{v}$, that contains only one degree-$4$ check node (the rest are degree-$2$ or -$1$) 
can be characterized by a $dot_m$ expansion ($2 \leq m \leq d_\mathrm{v}$) applied to a NETS substructure $\mathcal{S}'$ containing only one degree-$3$ check node in the $(a-1,b+2m-d_\mathrm{v})$ class. 
From $m$ edges connecting the variable node in $\mathcal{S} / \mathcal{S}'$  to $\mathcal{S}'$, one and $m-1$ are connected to degree-$3$ and degree-$1$ check nodes of $\mathcal{S}'$, respectively.
\end{pro}


\begin{cor}
\label{cor:nets4}
In variable-regular Tanner graphs with variable degree $d_\mathrm{v}$, all the NETSs containing only one degree-$4$ check node (the rest are degree-$2$ or -$1$) in the interest range of $a \leq a_{\max}$ and $b \leq b_{\max}$ ($a_{\max}$ less than the value in Table \ref{tab:disc} and $b_{max} \leq 4$) can be found by two successive applications of $dot_m$ expansions to all the ETSs  in the range of $a \leq a_{\max}-2$ and $b \leq b'_{\max}=b_{\max}+2d_\mathrm{v}-2$.
\end{cor}


\begin{pro}
\label{lem:nets34}
Any NETS structure $\mathcal{S}$  in the interest class of $(a,b)$ for variable-regular graphs with variable degree $d_\mathrm{v}$, that contains one degree-$4$ and one degree-$3$ check nodes (the rest are degree-$2$ or -$1$) 
can be characterized by a $dot_m$ expansion ($2 \leq m \leq d_\mathrm{v}$) applied to one of the following substructures $\mathcal{S}'$:
(i) a NETS substructure, containing only one degree-$3$ check node, in the $(a-1,b+2m-2-d_\mathrm{v})$ class, where out of $m \geq 2$ edges connecting the variable node in $\mathcal{S} / \mathcal{S}'$  to $\mathcal{S}'$, one is connected to a degree-$3$ check node, one to a degree-$2$ check node and $m-2$ to degree-1 check nodes of $\mathcal{S}'$.
(ii) a NETS substructure, containing two degree-$3$ check nodes,  in the $(a-1,b+2m-d_\mathrm{v})$ class, where out of $m \geq 2$ edges connecting the variable node in $\mathcal{S} / \mathcal{S}'$  to $\mathcal{S}'$, one and $m-1$ are connected to degree-$3$ and degree-$1$ check nodes of $\mathcal{S}'$, respectively.
\end{pro}


\begin{cor}
\label{cor:nets34}
In variable-regular Tanner graphs with variable degree $d_\mathrm{v}$, all the NETSs containing one degree-$4$ and one degree-$3$ check nodes (the rest are degree-$2$ or -$1$) in the interest range of $a \leq a_{\max}$ and $b \leq b_{\max}$ ($a_{\max}$ less than the value in Table \ref{tab:disc} and $b_{max} \leq 4$) can be found by three successive applications of $dot_m$ expansions to all the ETSs in the range of $a \leq a_{\max}-2$ and $b \leq b'_{\max}= b_{\max}+3d_\mathrm{v}-4$.
\end{cor}



\begin{rem}
Note that in all cases discussed in Corollaries \ref{cor:nets3}-\ref{cor:nets34}, successive applications of $dot_m$ expansions to the ETSs can result in finding structures in addition to the ones that are of interest in these corollaries. 
\end{rem}

Corollaries~\ref{cor:nets3}-\ref{cor:nets34} demonstrate that by increasing the multiplicity of check nodes with degrees larger than $2$ and the degrees of such check nodes, the range of $b$ values for ETSs that are needed to provide an exhaustive search of such NETSs is increased. To have an efficient NETS search algorithm based on successive $dot_m$ expansions of ETSs, we limit the multiplicity and the degrees of such check nodes to the following cases in the rest of this paper:
NETSs containing at most four degree-$3$ check nodes, or containing only one degree-$4$ and at most one degree-$3$ check nodes.
We use notations $N_{3}, N_{3,3}, N_{3,3,3}$, and $N_{3,3,3,3}$, 
to denote NETS structures with only one up to four check nodes of degree $3$. Notations $N_4$ and $N_{4,3}$ are used for NETS structures that contain only one degree-$4$ check node and those with only one degree-$4$ and one degree-$3$ check nodes, respectively. 

\begin{cor}
\label{cor:netsall}
In variable-regular Tanner graphs with variable degree $d_\mathrm{v}$, all the $N_{3}$ ,$N_{3,3}$ ,$N_{3,3,3}$ ,$N_{3,3,3,3}$, $N_4$, and $N_{4,3}$ in the interest range of $a \leq a_{\max}$ and $b \leq b_{\max}$ ($a_{\max}$ less than the value in Table \ref{tab:disc} and $b_{max} \leq 4$) can be found by up to four successive applications of $dot_m$ expansions to all the ETSs in the range of $a \leq a_{\max}-1$ and $b \leq b'_{\max}= b_{\max}+\max\{(3d_\mathrm{v}-4),(4d_\mathrm{v}-8)\}$.
\end{cor}

By restricting the NETS structures to those discussed above, we limit the maximum size $a_{max}$ of NETSs that can be exhaustively covered.  The following theorem provides the value of $a_{max}$ for Tanner graphs with different 
$d_\mathrm{v}$ and $g$ values. 
\begin{theo}
\label{pro:netsexh}
For a variable-regular Tanner graph with variable-degree  $d_\mathrm{v}$ and girth $g$, consider the union of sets $N_{3}$ ,$N_{3,3}$ ,$N_{3,3,3}$ ,$N_{3,3,3,3}$, $N_4$, and $N_{4,3}$, obtained by 
successive applications of $dot_m$  expansions ($m\geq 2$) to ETSs within the range indicated in Corollary~\ref{cor:netsall}. For $d_\mathrm{v}=3,4,5,6$, and $g=6,8,10$, Table \ref{tab:lowernets} provides the value of $a_{\max}$ such that
such a union gives an exhaustive list of NETSs of the Tanner graph within the range of $a \leq a_{max}$ and $b \leq b_{max}=4$.
\end{theo}

\begin{proof}
Based on the sets of NETSs that are covered, it is easy to see that the exhaustive search is limited by the size of the smallest structure in
sets $N_{3,3,3,3,3}$, $N_{4,3,3}$, $N_{4,4}$ and $N_5$. The structures in  $N_{3,3,3,3,3}$, however, have $b \geq 5$, and thus not in the range of interest of the theorem.
We thus find the size $a_{4,3,3}^*$, $a_{4,4}^*$ and $a_5^*$ (or a lower bound on the size) of the smallest structure in sets $N_{4,3,3}$, $N_{4,4}$ and $N_5$, respectively, and list $a_{max}=a^*-1$ in Table~\ref{tab:lowernets},
where $a^* = \min\{a_{4,3,3}^*, a_{4,4}^*, a_5^*\}$. 

For structures in $N_5$, we use Theorem \ref{lowf} with different values of $b \leq 4$, and choose the smallest lower bound as $a_5^*$.
For structures in $N_{4,3,3}$ and $N_{4,4}$, we use 
the tree-like expansion of the NETS structure as in Fig.~\ref{treegen}, starting from a degree-$4$ check node at the root in $L_1$. The tree thus has four variable nodes in $L_2$. The idea is to grow this tree into 
a NETS structure of smallest size with no cycle of length smaller than $g$ and with the given $b$ value, where out of $b$ unsatisfied check nodes in the case of $N_{4,3,3}$, two of them have degree $3$. 
To minimize the size, one needs to select the check nodes to have the minimum degree within the above constraints. For structures in $N_{4,3,3}$, this means selecting all the satisfied check nodes (other than the root) to have degree $2$ and
all the $b-2$ unsatisfied check nodes to have degree $1$. For structures in $N_{4,4}$, it means that all the satisfied check nodes, except for the root and one other check node with degree $4$, the rest must have degree $2$.
The $b$ unsatisfied check nodes in this case all have degree $1$.
To satisfy the girth constraint, all the variable
and check nodes in the first $g/2$ layers of the tree must be distinct (i.e., no cycle should appear in the subgraph). Moreover, in the tree, there are four subgraphs, each starting from one variable node
at $L_2$. To avoid having cycles shorter than $g$ in these subgraphs, any new variable (check) node at $L_{g/2+1}$, for $g/2$ odd (even), can only be connected to the check (variable) nodes of each such
subgraph at most once. Therefore, for odd values of $g/2$, at $L_{g/2+1}$, we need, at least, as many variable nodes as the number of edges emanating from the check nodes at $L_{g/2}$ of each subgraph to $L_{g/2+1}$. 
Also, for even values of $g/2$, if the number of variable nodes in a subgraph at $L_{g/2}$ times $d_\mathrm{v}-1$ is larger than the number of the rest of variable nodes at $L_{g/2}$ (in the other $k-1$ subgraphs), 
more variable nodes are needed to be added at $L_{g/2+2}$ to complete the connections required for check nodes at $L_{g/2+1}$.

Considering the above constraints, for both cases of structures in $N_{4,3,3}$ and $N_{4,4}$, and for each value of $b$, we find the structure with the smallest number of variable nodes. The values 
 $a_{4,3,3}^*$ and $a_{4,4}^*$ are then obtained by taking the minimum among the smallest sizes corresponding to five different values of $b=0, \ldots, 4$. In the following, we discuss in more details, 
 the proof for one entry of Table~\ref{tab:lowernets}. Proofs for other entries are similar.

Consider Tanner graphs with $d_\mathrm{v}=3$, $g=8$ and NETSs with $b_{\max}=4$. Based on Theorem \ref{lowf}, we have $a_5^*=12$. 

For $N_{4,3,3}$, to minimize the size of a NETS, the two degree-$3$ check nodes must be located at $L_3$ and be connected to two different variable nodes at $L_2$. This minimizes the number of variable nodes needed in
the higher layers of the tree while satisfying the girth constraint.
All the remaining $b-2$ unsatisfied check nodes with degree $1$ must also be located at $L_3$. 
The remaining check nodes at $L_3$ are thus  
$8-b$ degree-$2$ check nodes. This means there must be $4+(8-b)$ variable nodes at $L_4$, and $4+4+8-b=16-b$ variable nodes in the whole structure up to $L_4$. 
It appears that by proper addition of check nodes in $L_5$, no more variable node is needed in $L_6$. The smallest size of structures in $N_{3,3,4}$ for different values of $b \leq 4$ is thus obtained by $b=4$, and we 
have $a_{4,3,3}^* = 12$. As two examples, the smallest NETS structures for $b=3$ and $b=4$ are given in Fig.~\ref{smallNETS}.

For $N_{4,4}$, to minimize the size of NETS, the second degree-$4$ check node must be located at $L_3$.  
All the $b$ degree-$1$ check nodes must also be at $L_3$.  Out of $4 (d_\mathrm{v}-1)=8$ check nodes in $L_3$, one is
degree-$4$, $b$ are degree-$1$ and $7-b$ are degree-$2$. This means there are $3+(7-b)$ variable nodes at $L_4$. 
If $b \geq 1$, to satisfy the girth constraint with minimum number of variable nodes, one degree-$1$ check
node in $L_3$ is connected to the same variable node in $L_2$ that has also a connection to the degree-$4$ check node in $L_3$. 
If $b>1$, the rest of degree-$1$ check node(s) are each connected to another (different) variable node in $L_3$. 
Now, for $b > 1$, consider a variable node $v_1$ in $L_2$ that is connected to one degree-$4$ and one degree-$1$ check node in $L_3$ and call the subtree rooted at $v_1$ as subtree $1$. 
This subtree has $3$ variable nodes at $L_4$ that must be connected to $3 (d_\mathrm{v}-1)=6$ distinct check nodes at $L_5$. To complete the connections of these check nodes,
at least six variable nodes should exist at $L_4$ of the rest of the subtrees (excluding subtree $1$), otherwise,
more variable nodes are needed at $L_6$. By considering all the cases of $b \leq 4$, we conclude that $a_{4,4}^* = 13$. As two examples, the smallest NETS structures for $b=1$ and $b=4$ are shown in Fig.~\ref{smallNETS}.

Based on the above, for $d_\mathrm{v}=3$, $g=8$, we have $a_{\max} = \min\{12,12,13\} - 1 =11$.
\end{proof}

\begin{table*}[]
\centering
\renewcommand{\arraystretch}{1.1}
\setlength{\tabcolsep}{0.5pt}
\caption{The maximum size $a_{\max}$ of $(a,b)$ NETSs that can be searched exhaustively within the range $b \leq b_{max}=4$ by successive applications of $dot_m$ expansions to $(a',b')$ ETSs with size up to $a_{max}-1$ and $b' \leq b'_{\max}$. (The lower bound on the size of smallest possible NETS with $b \leq 4$ is given in  brackets.)}
\label{tab:lowernets}
\begin{tabular}{||c|c|c|c||c|c|c||c|c|c||c|c|c||}
\cline{1-13}
&\multicolumn{3}{c||}{$d_\mathrm{v}=3$}&\multicolumn{3}{c||}{$d_\mathrm{v}=4$}&\multicolumn{3}{c||}{$d_\mathrm{v}=5$}&\multicolumn{3}{c||}{$d_\mathrm{v}=6$}\\
\cline{2-13}
&$g=6$&$g=8$&$g=10$&$g=6$&$g=8$&$g=10$&$g=6$&$g=8$&$g=10$&$g=6$&$g=8$&$g=10$\\
\cline{1-13}
$a_{\max}$&$\begin{array}{@{}c@{}} 6(4) \end{array} $&$\begin{array}{@{}c@{}} 11(6) \end{array} $&$\begin{array}{@{}c@{}} 16(8) \end{array} $&$\begin{array}{@{}c@{}} 6(5) \end{array} $&$\begin{array}{@{}c@{}}15(9) \end{array} $&$\begin{array}{@{}c@{}}24(15) \end{array} $&$\begin{array}{@{}c@{}}8(6) \end{array} $&$\begin{array}{@{}c@{}}19(12) \end{array} $&$\begin{array}{@{}c@{}}43(24) \end{array} $&$\begin{array}{@{}c@{}}8(7) \end{array} $&$\begin{array}{@{}c@{}}21(15) \end{array} $&$\begin{array}{@{}c@{}}46(35) \end{array} $\\
\cline{1-13}
$b'_{\max}$&$\begin{array}{@{}c@{}} 9 \end{array} $&$\begin{array}{@{}c@{}} 9 \end{array} $&$\begin{array}{@{}c@{}} 9 \end{array} $&$\begin{array}{@{}c@{}} 12 \end{array} $&$\begin{array}{@{}c@{}}12 \end{array} $&$\begin{array}{@{}c@{}}12 \end{array} $&$\begin{array}{@{}c@{}}16 \end{array} $&$\begin{array}{@{}c@{}}16 \end{array} $&$\begin{array}{@{}c@{}}16 \end{array} $&$\begin{array}{@{}c@{}}20 \end{array} $&$\begin{array}{@{}c@{}}20 \end{array} $&$\begin{array}{@{}c@{}}20 \end{array} $\\
\cline{1-13}
\end{tabular}
\end{table*}

\begin{figure}[] 
\centering
\includegraphics [width=0.6\textwidth]{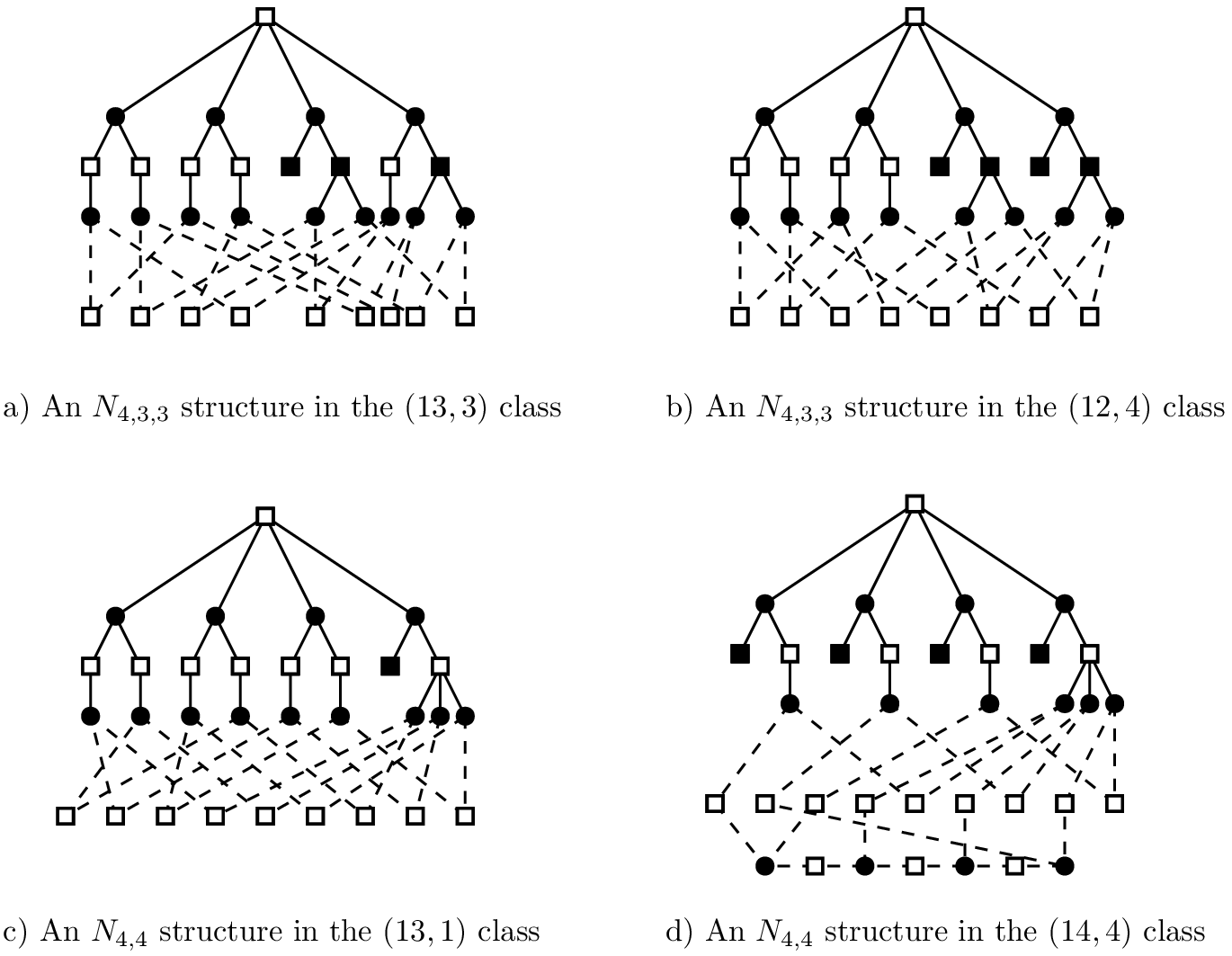}
\caption{Four examples of smallest possible NETS structures in $N_{4,3,3}$ and $N_{4,4}$ for variable-regular graphs with $d_\mathrm{v}=3$, $g=8$, in the range $b \leq 4$.}
\label{smallNETS}
\end{figure}

Using Corollary \ref{cor:netsall}, one can find the $b'_{\max}$ value which indicates the range of $b$ values for ETSs that are required for the exhaustive search of the desired NETSs. The $b'_{\max}$ values for graphs with different $d_\mathrm{v}$ values are also provided in Table \ref{tab:lowernets}. In Table \ref{tab:lowernets}, we have also included the lower bound on the size of the smallest possible NETS with $b \leq 4$ in brackets. As an example, the entries corresponding to $d_\mathrm{v}=3$, and $g=8$ in Table \ref{tab:lowernets} show that, for such variable-regular graphs, we can exhaustively search all the NETSs with $a=6,7,8,9,10,11$ and $b \leq 4$.
 
The pseudo-code of the proposed search algorithm is given in Algorithm \ref{algnets}.  In the proposed \textit{dot-based} NETS search algorithm, the input is the exhaustive 
list of  ETSs  in the range of $a \leq a_{max}-1$ and $b'_{max}$. In the search process, $dot$ expansion is applied to any instance of TSs (ETSs and NETSs) in the interest range of $a\leq a_{max}-1$ and $b \leq b'_{max}$. 
The sets ${\cal I}_{TS}^{a}$ and ${\cal I}_{ETS}^{a}$  are the sets of all the instances of TSs and ETSs in the $(a,b)$ classes with $b \leq b'_{\max}$, respectively. The set ${\cal I}_{NETS}^{a}$ is the set of all the instances of NETSs in the $(a,b)$ classes with $b \leq b_{\max}$.

\begin{algorithm}
\caption{{\bf (NETS Search)} Finds list of the instances of $(a,b)$ NETS structures of a variable-regular Tanner graph $G=(V,E)$ with girth $g$ and variable degree $d_\mathrm{v} $, for $a \leq a_{\max}$ and $b \leq b_{max}$ ($a_{\max}$
is obtained from Table~\ref{tab:lowernets} for $b_{max}=4$). The input is all the instances of $(a,b)$ ETS structures, in the range $a \leq a_{\max}-1$ and $b \leq b'_{\max}$, $\mathcal{I}_{ETS}$ ($b'_{\max}$
is obtained from Table~\ref{tab:lowernets}).  The output is the set ${\cal I}_{NETS}$, which contains  the instances of NETSs in the interest range. ${\cal I}_{NETS}$= \textbf{NetsSrch} (${\cal I}_{ETS}, a_{max}, b_{max}$)}
\label{algnets}
 \begin{algorithmic} [1] 
\State \textbf{Inputs:} $G, g, d_\mathrm{v}, \mathcal{I}_{ETS}$ ($a_{max}$, $b_{max}=4$).
\State  \textbf{Initializations:}  $ {\cal I}_{NETS}^{a} \gets \emptyset; {\cal I}_{TS}^{a}={\cal I}_{ETS}^{a}, ~\forall ~a \leq a_{\max}$; $b'_{\max}$ is obtained from Table \ref{tab:lowernets}.
\For {$a=g/2, \dots, a_{\max}-1$}
\For {any structure $\mathcal{S} \in {\cal I}_{TS}^{a}$}
\parState {Consider $\mathcal{V}$ to be the set of variable nodes in $V \setminus \mathcal{S}$  which have at least two connections to the check nodes in $\Gamma{(\mathcal{S})}$.} \label{dotnets}
\vspace{1pt}
\For {each variable node $v \in \mathcal{V}$}
\parState{$\mathcal{S}'=\{\mathcal{S} \cup v\} \setminus {\cal I}_{TS}^{a+1}$.}
\parState{$b=|\Gamma_{o}{(\mathcal{S}')}|$.}
\If {$b \leq b'_{\max}$}
\parState {${\cal I}_{TS}^{a+1}= {\cal I}_{TS}^{a+1} \cup \mathcal{S}'$.}
\If {$b \leq b_{\max}$}
\parState {${\cal I}_{NETS}^{a+1}= {\cal I}_{NETS}^{a+1} \cup \mathcal{S}'$.}
\EndIf
\EndIf
\EndFor
\EndFor
\EndFor
\State \textbf{Output:} ${\cal I}_{NETS}= \{{\cal I}_{NETS}^{a},~\forall ~a \leq a_{\max}\}$.
\end{algorithmic}
\end{algorithm}



\begin{rem}
\label{rem:exhau}
We note that if in Algorithm \ref{algnets}, we increase the value of $a_{max}$ beyond that of Table \ref{tab:lowernets} (but less than the one in Table \ref{tab:disc}), 
by exhaustive search of ETSs in the range of $a \leq a_{max}$ and $b \leq b'_{max}$, we can still find all the $N_{3},N_{3,3},N_{3,3,3},N_{3,3,3,3},N_{4},N_{4,3}$ structures in the new range of $a \leq a_{max}$ and $b \leq 4$, 
but there is no guarantee to find the other NETS structures in the new range exhaustively.
\end{rem}

\subsection{Non-Exhaustive Search  of NETSs in Variable-Regular LDPC Codes}
\label{sec:non-ex}

The  exhaustive search of NETSs proposed in Subsection \ref{sec:nets} has two limitations. First, the value of $b'_{\max}$ obtained in Subsection \ref{sec:nets}, is rather large which implies a high complexity for the exhaustive search of ETSs. 
Moreover, for the given values of $d_\mathrm{v}$, $g$ and  $b_{\max}$, the value of $a_{\max}$ is relatively small. For these two reasons, we  propose a non-exhaustive search of NETSs in a wider range of $a$ and $b$ values based on setting $b'_{\max}=b_{\max}+t$, where $t\geq 1$, instead of the value indicated in Table~\ref{tab:lowernets}.  Our experimental results show that by increasing $b'_{\max}$ beyond $b_{\max}+2$, 
the number of new NETSs that can be found in the interest range is negligible.
 
The NETS search algorithm proposed in Algorithm \ref{algnets} can also be used for the non-exhaustive search of NETSs. As the input, in this case, one should find and provide all the ETSs in the range $a \leq a_{max}$ and $b \leq b'_{max}$. 
However, since the $b'_{max}$ is less than the value given in Table \ref{tab:lowernets}, the list of NETSs, $\mathcal{I}_{NETS}$, would be non-exhaustive. One should also note that since the algorithm imposes no restriction 
on the degree of check nodes of searched NETSs, by increasing $a_{\max}$, some other NETSs with combination of different check node degrees can be found as well. 

\subsection{Search Algorithm to Find NETSs in Irregular LDPC Codes}
\label{sec:irreg}

Due to the variety of variable degrees in variable-irregular LDPC codes, we are not able to provide results similar to those in Subsection \ref{sec:nets} in relation to exhaustive search of NETSs in irregular codes.
Algorithm \ref{algnets} can, however, be still used for the non-exhaustive search of NETSs in irregular graphs. To obtain an exhaustive list of ETSs as the input to Algorithm \ref{algnets} in this case, one can use the search algorithms of \cite{hashemiireg}.
 
\section{Bounds on the Stopping Distance of LDPC Codes}
\label{sec:bounds}
Stopping sets can be viewed as a subset of TSs, where any check node has a degree of at least two. Elementary SSs (ESSs) and non-elementary SSs (NESSs) are thus subsets of ETSs and NETSs, respectively.
In the following, we tailor/modify the results established for ETSs and NETSs for ESSs and NESSs, respectively. 

\subsection{Lower Bound on the Stopping Distance of Variable-Regular LDPC Codes}
\label{sec:lowbound}

By definition, an ESS is a TS for which the degree of all the check nodes is $2$. Any  ESS thus corresponds to a LETS with $b=0$. The following lower bound on stopping distance is simple to prove.

\begin{pro}
\label{lem:smallpossi}
The result of Theorem~\ref{lowf} with $k=2$ and $b=0$ provides a lower bound, $L_{SS_1}$, on $s_{\min}$ for variable-regular LDPC codes.
\end{pro}


\begin{rem}
We note that the result of Proposition \ref{lem:smallpossi} is essentially the same as the lower bound obtained in \cite{Orlit}  on the stopping distance of variable-regular LDPC codes.
\end{rem}

To potentially improve the lower bound of Proposition~\ref{lem:smallpossi}, $L_{SS_1}$, we use the fact that ESSs, as a special case of LETSs, 
have a graphical structure that lends itself well to the efficient exhaustive \textit{dpl} search algorithm of~\cite{hashemireg}. 
Using the $dpl$ search algorithm with $b_{\max}=0$, we can efficiently and exhaustively find all the ESSs of a variable-regular LDPC code with a maximum given size $a_{\max}$.
In the following, we establish a lower bound,  $L_{SS_2}$ ($L_{SS_2} > L_{SS_1}$) , on the size of smallest NESSs. We then perform an exhaustive $dpl$-based search of ESSs of maximum size $a_{\max} = L_{SS_2} -1$.
If this search does not find any ESS, then we establish $L_{SS_2} \leq s_{\min}$. Otherwise, the smallest size of found ESSs is the exact value of $s_{\min}$.

\begin{pro}
\label{lem:smallpossinon}
The result of Theorem~\ref{lowf} with $k=3$ and $b=1$ provides a lower bound, $L_{SS_2}$, on the size of  NESSs.
\end{pro}

To further improve the lower bound of Proposition~\ref{lem:smallpossinon} on $s_{\min}$, if possible, we need to perform an exhaustive search of NESSs. 
This can be performed, by using the NETS search algorithm of Section \ref{sec:searTS} with some modifications as described below.

We first note that, compared to Subsection \ref{sec:nets}, here, we are not interested in NETSs with unsatisfied check nodes of degree-$1$. This implies that the range of exhaustive search for NESSs, as a subset of NETSs, can be potentially increased.

We use notations $SS_{3}$, $SS_{3,3}$, $SS_{3,3,3}$, $SS_{3,3,3,3}$, to denote NESSs with only one up to four check nodes of degree $3$, respectively. 
Notations $SS_4$ and $SS_{4,3}$ are used for NESSs that contain only one degree-$4$ check node and only one degree-$4$ and one degree-$3$ check nodes, respectively. 
Similar to Subsection \ref{sec:nets}, we limit the search of NESSs to the  following configurations:  $SS_{3}$,$SS_{3,3}$, $SS_{3,3,3}$, $SS_{3,3,3,3}$, $SS_{4}$, and $SS_{4,3}$.
The following result is in parallel with Corollary~\ref{cor:netsall}. 

\begin{cor}
\label{cor:ssall}
In variable-regular Tanner graphs with variable degree $d_\mathrm{v}$, all the $SS_{3}$,$SS_{3,3}$, $SS_{3,3,3}$, $SS_{3,3,3,3}$, $SS_{4}$, and $SS_{4,3}$ in the interest range of $a \leq a_{\max}$ ($a_{\max}$ less than the value in Table \ref{tab:disc}) can be found by up to four successive applications of $dot_m$ expansions to all the LETSs in the range of $a \leq a_{\max}-1$ and $b \leq b'_{\max}= 4+\max\{(3d_\mathrm{v}-4),(4d_\mathrm{v}-8)\}$.
\end{cor}

The following result (parallel to Theorem~\ref{pro:netsexh}) provides the range in which the NESS search is exhaustive.

\begin{theo}
\label{pro:nessexh}
For a variable-regular Tanner graph with variable-degree  $d_\mathrm{v}$ and girth $g$, consider the union of sets $SS_{3}$,$SS_{3,3}$, $SS_{3,3,3}$, $SS_{3,3,3,3}$, $SS_{4}$, and $SS_{4,3}$, obtained by 
successive applications of $dot_m$  expansions ($m\geq 2$) to LETSs within the range indicated in Corollary~\ref{cor:ssall}. For $d_\mathrm{v}=3,4,5,6$, and $g=6,8,10$, Table \ref{tab:lowerSSs} provides the value of $a_{\max}$ such that
such a union gives an exhaustive list of NESSs of the Tanner graph within the range of $a \leq a_{max}$.
\end{theo}

\begin{table*}[]
\centering
\renewcommand{\arraystretch}{1.1}
\setlength{\tabcolsep}{1pt}
\caption{The maximum size $a_{\max}$ of SSs that can be searched exhaustively by successive applications of $dot_m$ expansions to $(a',b')$ LETSs with size up to $a_{max}-1$ and $b' \leq b'_{\max}$.}
\label{tab:lowerSSs}
\begin{tabular}{||c|c|c|c||c|c|c||c|c|c||c|c|c||}
\cline{1-13}
&\multicolumn{3}{c||}{$d_\mathrm{v}=3$}&\multicolumn{3}{c||}{$d_\mathrm{v}=4$}&\multicolumn{3}{c||}{$d_\mathrm{v}=5$}&\multicolumn{3}{c||}{$d_\mathrm{v}=6$}\\
\cline{2-13}
&$g=6$&$g=8$&$g=10$&$g=6$&$g=8$&$g=10$&$g=6$&$g=8$&$g=10$&$g=6$&$g=8$&$g=10$\\
\cline{1-13}
$a_{\max}$&$\begin{array}{@{}c@{}} 7 \end{array} $&$\begin{array}{@{}c@{}} 13 \end{array} $&$\begin{array}{@{}c@{}} 19 \end{array} $&$\begin{array}{@{}c@{}} 8 \end{array} $&$\begin{array}{@{}c@{}}17 \end{array} $&$\begin{array}{@{}c@{}}29 \end{array} $&$\begin{array}{@{}c@{}}8 \end{array} $&$\begin{array}{@{}c@{}}19 \end{array} $&$\begin{array}{@{}c@{}}43 \end{array} $&$\begin{array}{@{}c@{}}10 \end{array} $&$\begin{array}{@{}c@{}}22 \end{array} $&$\begin{array}{@{}c@{}}56 \end{array} $\\
\cline{1-13}
$b'_{\max}$&$\begin{array}{@{}c@{}} 9 \end{array} $&$\begin{array}{@{}c@{}} 9 \end{array} $&$\begin{array}{@{}c@{}} 9 \end{array} $&$\begin{array}{@{}c@{}} 12 \end{array} $&$\begin{array}{@{}c@{}}12 \end{array} $&$\begin{array}{@{}c@{}}12 \end{array} $&$\begin{array}{@{}c@{}}16 \end{array} $&$\begin{array}{@{}c@{}}16 \end{array} $&$\begin{array}{@{}c@{}}16 \end{array} $&$\begin{array}{@{}c@{}}20 \end{array} $&$\begin{array}{@{}c@{}}20 \end{array} $&$\begin{array}{@{}c@{}}20 \end{array} $\\
\cline{1-13}
\end{tabular}
\end{table*}

\begin{proof}
The proof is similar to the proof of Theorem \ref{pro:netsexh} with the following differences: ($i$) Despite the case of Theorem \ref{pro:netsexh}, where NETSs with only $b \leq 4$ were studied, here, for NESSs, there is no such limitation, and we are interested in NESSs with any value of $b$ (including those with $b>4$), ($ii$) Since there exists no degree-$1$ check node in a SS, the degree of all the unsatisfied check nodes of NESSs should be odd values greater than or equal to $3$,
($iii$) Due to ($i$), in addition to minimum-size structures in $SS_{4,3,3}$, $SS_{4,4}$, $SS_{5}$, one needs to also consider the minimum-size structures in $SS_{3,3,3,3,3}$ as potentially limiting the range of exhaustive search, ($iv$) In the tree-like expansion of the subgraph, to minimize the size of the NESS and due to the non-existence of degree-$1$ check nodes, one needs to assume that except the few check nodes with degree larger than $2$, all the rest of check nodes have degree-$2$.
\end{proof}

Fig. \ref{smallNESS} shows four examples of smallest NESS structures of $SS_{4,3,3}$, $SS_{4,4}$, $SS_{5}$ and $SS_{3,3,3,3,3}$ for graphs with $d_\mathrm{v}=3$ and $g=8$.
\begin{figure}[] 
\centering
\includegraphics [width=0.6\textwidth]{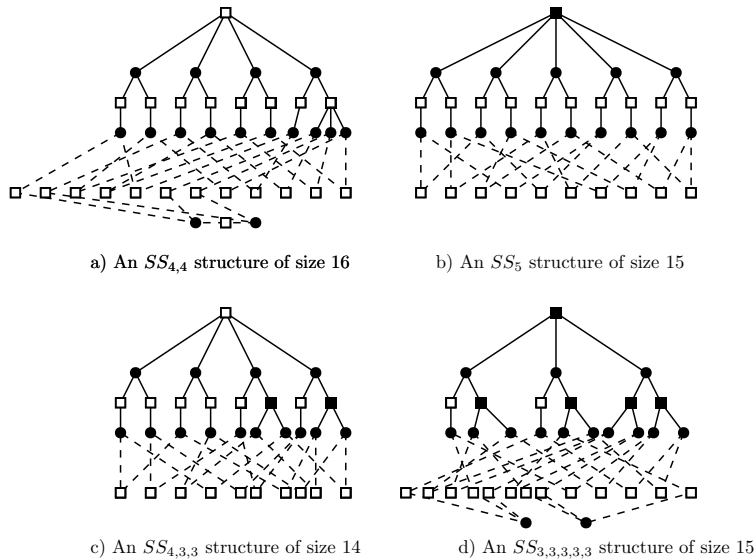}
\caption{Four examples of smallest NESS structures in variable-regular graphs with $d_\mathrm{v}=3,g=8$.}
\label{smallNESS}
\end{figure}

\begin{rem}
Note that the values of $a_{max}$ in Table \ref{tab:lowerSSs} are generally larger than those in Table \ref{tab:lowernets}. 
For example, while the range of exhaustive search of NESSs for variable-regular graphs with $d_\mathrm{v}=3$ and $g=8$ is $13$, this value for the exhaustive search of NETSs is $11$.
\end{rem}

If the exhaustive search of SSs (ESSs and NESSs) up to size $a_{max}$ listed in Table \ref{tab:lowerSSs} fails to find any SS, 
then $L_{SS_3} = a_{max}+1$ is a lower bound on $s_{\min}$. Otherwise, the smallest size of found SSs gives the exact value of $s_{\min}$. 

In Table \ref{tab:lowerboundsSS}, we have listed $L_{SS_3}$, as well as the values of $L_{SS_1}$ and $L_{SS_2}$, obtained from Propositions \ref{lem:smallpossi} and \ref{lem:smallpossinon}, for Tanner graphs with different values of $d_\mathrm{v}$ and $g$. 

\begin{table*}[]
\centering
\renewcommand{\arraystretch}{1.1}
\setlength{\tabcolsep}{0.7pt}
\caption{Values $L_{SS_3}$,$L_{SS_2}$ and $L_{SS_1}$  (as potential lower bounds on  $s_{min}$) for codes with $d_\mathrm{v}=3,4,5,6$, and $g=6,8,10$}
\label{tab:lowerboundsSS}
\begin{tabular}{||c|c|c|c||c|c|c||c|c|c||c|c|c||}
\cline{1-13}
&\multicolumn{3}{c||}{$d_\mathrm{v}=3$}&\multicolumn{3}{c||}{$d_\mathrm{v}=4$}&\multicolumn{3}{c||}{$d_\mathrm{v}=5$}&\multicolumn{3}{c||}{$d_\mathrm{v}=6$}\\
\cline{2-13}
&$g=6$&$g=8$&$g=10$&$g=6$&$g=8$&$g=10$&$g=6$&$g=8$&$g=10$&$g=6$&$g=8$&$g=10$\\
\cline{1-13}
$L_{SS_3}$&$\begin{array}{@{}c@{}} 8 \end{array} $&$\begin{array}{@{}c@{}} 14 \end{array} $&$\begin{array}{@{}c@{}} 20 \end{array} $&$\begin{array}{@{}c@{}} 9 \end{array} $&$\begin{array}{@{}c@{}} 18 \end{array} $&$\begin{array}{@{}c@{}} 30 \end{array} $&$\begin{array}{@{}c@{}} 9 \end{array} $&$\begin{array}{@{}c@{}} 20 \end{array} $&$\begin{array}{@{}c@{}} 44 \end{array} $&$\begin{array}{@{}c@{}} 11 \end{array} $&$\begin{array}{@{}c@{}} 23 \end{array} $&$\begin{array}{@{}c@{}} 57 \end{array} $\\
\cline{1-13}
$L_{SS_2}$&$\begin{array}{@{}c@{}} 5 \end{array} $&$\begin{array}{@{}c@{}} 9 \end{array} $&$\begin{array}{@{}c@{}} 13 \end{array} $&$\begin{array}{@{}c@{}} 7 \end{array} $&$\begin{array}{@{}c@{}} 12 \end{array} $&$\begin{array}{@{}c@{}} 21 \end{array} $&$\begin{array}{@{}c@{}} 7 \end{array} $&$\begin{array}{@{}c@{}} 15 \end{array} $&$\begin{array}{@{}c@{}} 31 \end{array} $&$\begin{array}{@{}c@{}} 8 \end{array} $&$\begin{array}{@{}c@{}} 18 \end{array} $&$\begin{array}{@{}c@{}} 43 \end{array} $\\
\cline{1-13}
$L_{SS_1}$&$\begin{array}{@{}c@{}} 4 \end{array} $&$\begin{array}{@{}c@{}} 6 \end{array} $&$\begin{array}{@{}c@{}} 10 \end{array} $&$\begin{array}{@{}c@{}} 5 \end{array} $&$\begin{array}{@{}c@{}} 8 \end{array} $&$\begin{array}{@{}c@{}} 17 \end{array} $&$\begin{array}{@{}c@{}} 6 \end{array} $&$\begin{array}{@{}c@{}} 10 \end{array} $&$\begin{array}{@{}c@{}} 26 \end{array} $&$\begin{array}{@{}c@{}} 7 \end{array} $&$\begin{array}{@{}c@{}} 12 \end{array} $&$\begin{array}{@{}c@{}} 37 \end{array} $\\
\cline{1-13}
\end{tabular}
\end{table*}

\begin{ex}
In variable-regular graphs with $d_\mathrm{v}=3$ and $g=8$, while the size of smallest possible ESSs (from Proposition \ref{lem:smallpossi}) is $6$, by the exhaustive search of ESSs, one can potentially improve the lower bound on $s_{min}$ to $9$. Moreover, by considering NESSs, the bound can be potentially improved further to $14$. 
\end{ex}

The pseudo code for obtaining a lower bound $s_{\min}^{(l)}$ on $s_{\min}$ is presented in Algorithm \ref{alg2}. The algorithm starts by exhaustively searching for ESSs of size at most $L_{SS_3}-1$ in Lines \ref{startESS}-\ref{endESS}. 
During the search of ESSs, if any ESS is found, the size of that ESS is assigned as the temporary value for $s_{min}^{(l)}$. (Since the search of LETSs is hierarchical, if any SS is found in Line~\ref{ss lets}, it is the smallest one in the range of interest.) 
If this temporary $s_{min}^{(l)}$ is larger than $L_{SS_2}$, then the NETS search algorithm of Subsection \ref{sec:nets} (Algorithm \ref{algnets}) 
is used to find NESSs with size less than this temporary $s_{min}^{(l)}$. If such a NESS is found, the size of that NESS  is assigned as the new and final value for $s_{min}^{(l)}$. (Since the search of NETSs is hierarchical, if any NESS is found in Line~\ref{ss nets}, it is the smallest one in the range of interest.)

\begin{algorithm}
\caption{Finding a lower bound $s_{\min}^{(l)}$ on the stopping distance of a variable-regular Tanner graph $G$ with variable degree $d_\mathrm{v}$ and girth $g$.}
\label{alg2}
 \begin{algorithmic} [1] 
\State \textbf{Inputs:} $G, g, d_\mathrm{v}$.
\parState {\textbf{Initializations:} Set $a_{\max}=L_{SS_3}$, and select  $b'_{\max}$ from Table \ref{tab:lowerSSs}. }
\State $s_{\min}^{(l)}=a_{\max}$.
\State{Run the exhaustive LETS search algorithm in the range of $a \leq a_{max} - 1$ and $b \leq b'_{max}$.} \label{runLETS}
\While{the $dpl$ search is running}  \label{startESS}
\If{a LETS in an $(a,0)$ class is found} \label{ss lets}
\State{Stop the search, and set $s_{\min}^{(l)}= a$, $a_{\max}=a$.} \label{remove1}
\EndIf
\EndWhile  \label{endESS}
\If{$s_{\min}^{(l)} > L_{SS_2}$}
\State{${\cal I}_{NETS}$= \textbf{NetsSrch}(${\cal I}_{LETS}$, $a_{max}$, $b_{max}=4$)} \label{runNETS}
\While{the NETS search is running} 
\If{a NESS of size $a$ is found, where $a < a_{\max}$,} \label{ss nets}
\State{Stop the search, and set $s_{\min}^{(l)}= a$.} \label{remove2}
\EndIf
\EndWhile 
\EndIf
\State \textbf{Output:} $s_{\min}^{(l)}$. 
\end{algorithmic}
\end{algorithm}

\begin{rem}
In Subsection \ref{sec:nets}, the input of Algorithm~\ref{algnets} for the search of NETSs in the range of $a\leq a_{max}$ and $b \leq 4$ was the list of all ETSs in the range of $a\leq a_{max}-1$ and $b \leq b'_{max}$. However,  NESSs are leafless, \textit{i.e.}, each variable node is connected to at least two check nodes. Therefore, for finding the NESSs in Line \ref{ss nets} of Algorithm \ref{alg2}, the input is just the list of LETSs in the range of $a\leq a_{max}-1$ and $b \leq b'_{max}$,  that has already been found in Line \ref{runLETS}.
\end{rem}


\begin{rem}
We note that by removing the conditions that stop the algorithm when an ESS or a NESS is found, one can find the list of all stopping sets with size less than or equal to $a_{max} = L_{SS_3} - 1$ exhaustively.
\end{rem}

\subsection{Upper Bound on the Stopping Distance of LDPC Codes}
\label{sec:upper}

If we fail to find the exact stopping distance of an LDPC code (variable-regular) based on the approach described in Subsection~\ref{sec:lowbound}, then, we have established that $s_{\min} \geq L_{SS_3}$. 
For such cases, in this subsection, we also find an upper bound on $s_{\min}$. To obtain this upper bound, we find a stopping set with size larger than or equal to $L_{SS_3}$.
We do this by devising a non-exhaustive search algorithm for SSs with a range that can go well beyond $L_{SS_3}$. This search algorithm is also applicable to irregular LDPC codes and can provide an upper bound on 
$s_{\min}$ for such codes.

The new algorithm also searches for both ESSs and NESSs, and to search for both categories, it requires to search for LETSs. The LETS search can be performed through the exhaustive $dpl$ searches of~\cite{hashemireg} and~\cite{hashemiireg}, for regular and irregular graphs, respectively. The problem, however, is that the complexity of such a search increases rather rapidly, if the range of search, indicated by the value of $a_{\max}$, is increased much beyond the value of $L_{SS_3}$. 
To overcome the problem of high complexity of the exhaustive $dpl$ search, in cases where the smallest size of stopping sets is well above $L_{SS_3}$, 
we modify the search such that it can handle larger values of $a_{\max}$. This however, comes at the expense of losing the exhaustiveness of the search, and thus we are not guaranteed to find 
the stopping sets with the lowest weight in our search.  In this part of the work, rather than selecting the $b'_{\max}$ as in the original $dpl$ characterization/search, we choose it to be a smaller value. 
To compensate for the detrimental effect that this new choice will have on the  exhaustiveness of the $dpl$ search, rather than using the specific expansion techniques that the original $dpl$ characterization 
determines for each LETS class, we apply \textit{all} the possible expansions from the set of $dot, path$ and $lollipop$ expansions to LETS structures in each class in the range of $a \leq a_{\max}$ and $b \leq b'_{\max}$. 
The only constraint for the application of a certain expansion technique to LETS structures within a specific class is that the expanded structure must still remain within the range $a \leq a_{\max}$ and $b \leq b'_{\max}$. 
Given the values $a_{\max}$ and $b'_{\max}$, Routine~\ref{algexp} provides a pseudo code for finding the list of expansion techniques that are required to be applied to all the LETSs in each $(a,b)$ class. 
These expansions are stored in the $(a,b)$ entry of table $\mathcal{EX}$, $\mathcal{EX}_{(a,b)}$.
The expansion $dot$ is applied to all the $(a,b)$ classes with $a \leq a_{\max}-1$. Also, $pa_m$ and $lo_m^c$ are applied to all the $(a,b)$ classes with $a \leq a_{\max}-m$.
The only constraint for using an expansion technique is that the $b$ value(s) of the new LETS structure(s) need to remain in the range identified by $b'_{\max}$.

\begin{routine}
 \caption{{\bf (Expansions)} Finds all the possible expansions $\mathcal{EX}_{(a,b)}$ for the $(a,b)$ classes of LETS structures, $g/2 \leq a < a_{\max}$,  $1 \leq b \leq b'_{\max}$, for a Tanner graph 
 with girth $g$ and variable degree $d_\mathrm{v}$ (for irregular graphs, $d_\mathrm{v}=d_{v_{\min}}$). $\mathcal{EX}=$Expansions$(a_{\max}, b'_{\max}, g, d_\mathrm{v})$}
\label{algexp}
 \begin{algorithmic} [1] 
\State \textbf{Initialization:} $a=a_{\max}-1$. 
\While {$a \geq g/2$}
\State  $\mathcal{EX}_{(a,b)} \gets dot,~\forall~b \leq b'_{\max}$.
\For {$b = 1, \ldots, b'_{\max}$}
\State $m=2$.
\While{$a+m \leq a_{\max}$}
\If {$b \leq b'_{\max}-2+m(d_\mathrm{v}-2)$}
\State $\mathcal{EX}_{(a,b)}\gets pa_m,lo_m^c$.
\EndIf
\State $m=m+1$.
\EndWhile
\EndFor
\State $a=a-1$.
\EndWhile
\State \textbf{Output:} $\mathcal{EX}$.
\end{algorithmic}
 \end{routine}

A pseudo code for obtaining an upper bound $s_{\min}^{(u)}$ on stopping distance is presented in Algorithm \ref{algup}. To start the algorithm, 
one can select $a_{\max}$ to be initially a rather large value $a_0$, say three or four times $L_{SS_3}$.
The procedures of searching for ESSs and NESSs are generally similar to those in Section \ref{sec:lowbound}, with some differences explained in the following. 
In Algorithm~\ref{alg2}, for the exhaustive search of LETSs (Line \ref{runLETS}), the set of expansions $\mathcal{EX}$, and $b'_{max}$ are obtained by the characterization algorithm  of \cite{hashemireg}.
Also, in Algorithm~\ref{alg2}, for the exhaustive search of NETSs in the range of $a \leq a_{max}$ and $b \leq 4$ (Line \ref{runNETS}), the value of $b'_{max}$ for the exhaustive 
search of LETSs is obtained from Table \ref{tab:lowerSSs}.
In Algorithm~\ref{algup}, however, for both non-exhaustive search of LETSs and NETSs, $b'_{\max}$ is chosen to be a rather small value. 
This value for variable-regular codes is set at $b'_{\max}=g/2(d_\mathrm{v}-2)$, in Algorithm \ref{algup}. This covers the class of shortest simple cycles of the graph. 
Also, for the search of NETSs (NESSs) in Algorithm \ref{algnets}, $b_{max}=b'_{max}$. For irregular graphs, the value is chosen as $b'_{\max}=4$ in Algorithm~\ref{algup}. 
Also, when the values of $a_{\max}$ and $b'_{\max}$ are set, the expansions $\mathcal{EX}$ needed for all the relevant classes of LETS structures are determined through Routine~\ref{algexp}.

If an ESS of size $a$ is found, then $a$ is a temporary upper bound for the stopping distance of the code. Then the NETS search is used to find any possible NESS with size less than the size of the smallest ESS.
If such a NESS is found, then its size is an upper bound on the stopping distance of the code. If the search terminated without finding any stopping set, or if one is interested in tightening the upper bound, one can increase the value of $b'_{\max}$ in a new search, to allow for covering more structures. In the latter case, where a stopping set of weight $s_{\min}^{(u)}$ has already been found, one should set 
$a_{\max} = s_{\min}^{(u)}-1$, for the new search. 

\begin{algorithm}
\caption{Finding an upper bound $s_{\min}^{(u)}$ on the stopping distance of Tanner graph $G$ with variable degree $d_\mathrm{v}$ and girth $g$.}
\label{algup}
 \begin{algorithmic} [1] 
 \State \textbf{Inputs:} $G, g, d_\mathrm{v}$.
\parState {\textbf{Initializations:} Set $a_{\max}= a_0$,  $b'_{\max}=g/2(d_\mathrm{v}-2)$ for variable-regular codes, or $b'_{\max}=4$ for irregular codes, and $s_{\min}^{(u)}= \infty$. (For irregular codes, $d_\mathrm{v}=d_{\mathrm{v}_{\min}}$.)}
\State{$\mathcal{EX}=${\bf Expansions}$(a_{\max}, b'_{\max}, g, d_\mathrm{v})$ (Routine~\ref{algexp}).}\label{psu:alg}
\State{Run the LETS search algorithm (\cite{hashemireg}, \cite{hashemiireg}) based on the expansions in $\mathcal{EX}$.}
\While{the LETS search is running} \label{startwhile}
\If{the run-time exceeds $T$,}
\State{Stop the LETS search, and go to Step \ref{end}.}
\EndIf
\If{a LETS in an $(a,0)$ class is found, where $a < s_{\min}^{(u)}$,}
\State{Stop the LETS search, set $s_{\min}^{(u)}= a, a_{max}=a$.}
\EndIf
\EndWhile \label{endwhile}
\State{${\cal I}_{NETS}$= \textbf{NetsSrch}(${\cal I}_{LETS},a_{\max},b'_{\max}$)}
\While{the NETS search is running} 
\If{a NESS of size $a$ is found, where $a < s_{\min}^{(u)}$,}
\State{Stop the NETS search,  set $s_{\min}^{(u)}= a$.}
\EndIf
\EndWhile
\If{$s_{\min}^{(u)}= \infty$,}
\State {$b'_{\max}=b'_{\max}+1$, and go to Step \ref{psu:alg}.}
\EndIf
\If{$s_{\min}^{(u)} < \infty$, but looking for a tighter bound,}\label{startif}
\State{$b'_{\max}=b'_{\max}+1$, $a_{\max}= s_{\min}^{(u)}-1$, and go to Step \ref{psu:alg}.}
\EndIf \label{endif}
\State \textbf{Output:} $s_{\min}^{(u)}$. \label{end}
\end{algorithmic}
\end{algorithm}

\section{Numerical results}
\label{sec:num} 
We have applied our technique to find lower and upper bounds on the stopping distance of a large number of variable-regular and irregular LDPC codes. 
These include both random and structured codes with a wide range of rates and block lengths. Here, we present the results for $20$ variable-regular and $8$ irregular codes.
These codes and their parameters can be seen in Tables \ref{tab:dss3} and \ref{tab:dssireg}, respectively.
For all the run-times reported in this paper, a desktop computer with $2.4$-GHz CPU and $8$-GB RAM is used, and the search algorithms are implemented in MATLAB.
In Tables \ref{tab:dss3} and \ref{tab:dssireg}, for the cases where the exact $s_{\min}$ is found, this value is reported in the 
column corresponding to the lower bound, and we have ``-'' in the upper bound column. Otherwise, the value $L_{SS_3}$ is reported as the lower bound, and the upper bound is obtained using the non-exhaustive $dpl$ search algorithm. In such cases, the value $b'_{\max}$ that has been used to provide the upper bound is reported in the last column of the table. 
For all cases, the run-time to obtain the lower and upper bounds are also reported.
For structured codes, their structural properties are used to simplify the search.
These codes are $\mathcal{C}_{10}, \mathcal{C}_{13}-\mathcal{C}_{20}$ in Table~\ref{tab:dss3}, and $\mathcal{C}_{21}-\mathcal{C}_{26}$, in Table~\ref{tab:dssireg}.
Also, for all cases, the letter \textit{e} or \textit{n} is reported in brackets to indicate whether the smallest SS  found in the search algorithm is elementary or non-elementary, respectively.

\begin{table}[]
\centering
\renewcommand{\arraystretch}{0.9}
\setlength{\tabcolsep}{1.2pt}
\caption{Stopping Distance or Bounds on the Stopping  Distance of some Variable-Regular LDPC Codes}
\label{tab:dss3}
\begin{tabular}{||c|c|c|c|c||c|c|c||}
\cline{1-8}
Code&$d_\mathrm{v}$&Girth&Rate&Length&Lower Bound&Upper Bound&$b'_{\max}$\\

\cline{1-8}
$\mathcal{C}_1$\cite{mackayencyclopedia}&3&6&0.5&504&$\begin{array}{@{}c@{}}s_{\min} \geq 7 \\ \hdashline 5~sec. \end{array} $&$\begin{array}{@{}c@{}} 16(\textit{n}) \\ \hdashline 5~min. \end{array} $&4\\
\cline{1-8}
$\mathcal{C}_2$\cite{mackayencyclopedia}&3&6&0.5&816&$\begin{array}{@{}c@{}}s_{\min}\geq 7 \\ \hdashline 6~sec. \end{array} $&$\begin{array}{@{}c@{}} 24(\textit{n}) \\ \hdashline 6~min. \end{array} $&4\\
\cline{1-8}
$\mathcal{C}_3$\cite{mackayencyclopedia}&3&6&0.5&1008&$\begin{array}{@{}c@{}}s_{\min}\geq 7 \\ \hdashline 6~sec. \end{array} $&$\begin{array}{@{}c@{}} 26(\textit{n}) \\ \hdashline 34~min. \end{array} $&5\\
\cline{1-8}
$\mathcal{C}_{4}$\cite{mackayencyclopedia}&3&6&0.77&1057&$\begin{array}{@{}c@{}} 7(\textit{n}) \\ \hdashline 43~sec. \end{array} $&-&-\\
\cline{1-8}
$\mathcal{C}_{5}$\cite{mackay2}&3&6&0.75&2000&$\begin{array}{@{}c@{}} 6(e) \\ \hdashline 15~sec. \end{array} $&-&-\\
\cline{1-8}
$\mathcal{C}_{6}$\cite{mackay2}&3&6&0.77&3000&$\begin{array}{@{}c@{}}s_{\min}\geq 7 \\ \hdashline 48~sec. \end{array} $&$\begin{array}{@{}c@{}} 8(e) \\ \hdashline 1~min. \end{array} $&4\\
\cline{1-8}
$\mathcal{C}_{7}$\cite{mackay2}&3&6&0.8&5000&$\begin{array}{@{}c@{}} 6(e) \\ \hdashline 28~sec. \end{array} $&-&-\\
\cline{1-8}
$\mathcal{C}_{8}$\cite{mackay2}&3&6&0.81&8000&$\begin{array}{@{}c@{}}s_{\min}\geq 7 \\ \hdashline 51~sec. \end{array} $&$\begin{array}{@{}c@{}} 14(e) \\ \hdashline 23~min. \end{array} $&4\\
\cline{1-8}
$\mathcal{C}_{9}$\cite{mackayencyclopedia}&3&6&0.87&16383&$\begin{array}{@{}c@{}}s_{\min}\geq 7 \\ \hdashline 209~sec. \end{array} $&$\begin{array}{@{}c@{}} 9(\textit{n}) \\ \hdashline 3~min. \end{array} $&3\\
\cline{1-8}
$\mathcal{C}_{10}$\cite{Tanner3}&3&8&0.41&155&$\begin{array}{@{}c@{}} s_{\min}\geq 13 \\ \hdashline 6~sec. \end{array} $&$\begin{array}{@{}c@{}} 18(\textit{n}) \\ \hdashline 45~sec. \end{array} $&4\\
\cline{1-8}
$\mathcal{C}_{11}$\cite{Peg}&3&8&0.5&504&$\begin{array}{@{}c@{}} s_{\min}\geq 13 \\ \hdashline 363~sec. \end{array} $&$\begin{array}{@{}c@{}} 19(\textit{n}) \\ \hdashline 10~min. \end{array} $&5\\
\cline{1-8}
$\mathcal{C}_{12}$\cite{Peg}&3&8&0.5&1008&$\begin{array}{@{}c@{}} s_{\min}\geq 13 \\ \hdashline 245~sec. \end{array} $&$\begin{array}{@{}c@{}} 37(\textit{n}) \\ \hdashline 35~min. \end{array} $&5\\
\cline{1-8}
$\mathcal{C}_{13}$\cite{tas2}&3&8&0.88&4000&$\begin{array}{@{}c@{}} 8(\textit{e}) \\ \hdashline 73~sec. \end{array} $&-&-\\
\cline{1-8}
$\mathcal{C}_{14}$\cite{Tanner4}&3&8&0.82&5219&$\begin{array}{@{}c@{}} 12(\textit{e}) \\ \hdashline 913~sec. \end{array} $&-&-\\
\cline{1-8}
$\mathcal{C}_{15}$\cite{tas}&3&10&0.5&546&$\begin{array}{@{}c@{}} 14(\textit{e}) \\ \hdashline 42~sec. \end{array} $&-&-\\
\cline{1-8}
$\mathcal{C}_{16}$\cite{Rosenthal}&3&12&0.5&4896&$\begin{array}{@{}c@{}} 24(\textit{e}) \\ \hdashline 729~sec. \end{array} $&-&-\\
\cline{1-8}
$\mathcal{C}_{17}$\cite{tas2}&4&8&0.69&1274&$\begin{array}{@{}c@{}} 8(\textit{e}) \\ \hdashline 10~sec. \end{array} $&-&-\\
\cline{1-8}
$\mathcal{C}_{18}$\cite{tas2}&4&8&0.77&2890&$\begin{array}{@{}c@{}} s_{\min}\geq 17 \\ \hdashline 468~sec. \end{array} $&$\begin{array}{@{}c@{}} 20(\textit{e}) \\ \hdashline 9~min. \end{array} $&10\\
\cline{1-8}
$\mathcal{C}_{19}$\cite{tas2}&5&8&0.23&210&$\begin{array}{@{}c@{}} 10(\textit{e}) \\ \hdashline 4~sec. \end{array} $&-&-\\
\cline{1-8}
$\mathcal{C}_{20}$\cite{tas2}&5&8&0.75&8000&$\begin{array}{@{}c@{}} s_{\min}\geq 19 \\ \hdashline 154~sec. \end{array} $&-&-\\
\cline{1-8}
\end{tabular}
\end{table}

We note that the lower bounds (or the exact stopping distances) are all obtained in times that are at most about $16$ minutes, and in many cases, only in a few seconds. 
The upper bounds are obtained in at most about $35$ minutes, and in many cases, less than $4$ minutes. 
 Using a computer with Core 2 Duo E6700 2.67GHz CPU and 2 GB of RAM, it took the search algorithm of~\cite{Hirotomo2}, about $600$ and $3085$ hours to provide an upper bound on the stopping distance of $\mathcal{C}_1$ and $\mathcal{C}_3$, respectively. In comparison, it has taken the non-exhaustive $dpl$ search algorithm of this paper only $5$ and $34$ minutes to find the same upper bounds  for $\mathcal{C}_1$ and $\mathcal{C}_3$, respectively.
Also the exact stopping distance of $\mathcal{C}_1$ has been reported in \cite{Rosnes}, which is matched with the bound reported here (the run-time has not been reported in \cite{Rosnes}).

To the best of our knowledge, the upper bound for random codes ${\cal C}_2$ and $\mathcal{C}_4$ has not been reported in the literature. 
It takes our algorithm only $43$ seconds to find the exact stopping distance of $\mathcal{C}_4$, a random high rate code with block length $1057$.

We  believe that the run-times reported here would be much less than those of any existing search algorithm. 
In fact, in our opinion, no existing algorithm would be able to handle $\mathcal{C}_{9}$, which is a code of rate $0.87$ and block length $16383$. It takes our algorithm only about $2$ and $3$ minutes to provide the lower and upper bounds of $7$ and $9$ on the $s_{min}$ of this code, respectively.

Codes $\mathcal{C}_{5}$-$\mathcal{C}_{8}$ are four high-rate random codes with variable degree $3$ and girth $6$ constructed by the program of \cite{mackay2}.\footnote{Using \textit{code6.c} with seed$=380$ in \cite{mackay2}.}
These random high-rate codes with large block lengths are challenging codes for all the existing approaches in the literature.
One can see that the exact stopping distance or the lower and upper bounds of these codes have been found by the proposed algorithms in most cases in a few seconds. 
To the best of our knowledge, except a few structured medium length codes with rate $0.5$, no result has been reported in the literature for codes with relatively large block length and high rate.

Also, an upper bound of $18$ on the stopping distance of $\mathcal{C}_{10}$ (Tanner $(155,64)$) has been found in just $45$ seconds. The obtained upper bound  matches  the exact value of $s_{min}$
reported in \cite{Rosnes}. Among seven variable-regular LDPC codes reported in \cite{Rosnes}, the run-time for finding the stopping distance of only two structured small block length codes, including $\mathcal{C}_{10}$, has been reported. 
The stopping distance of $\mathcal{C}_{10}$ has been found in about $1$ minute on a standard desktop computer \cite{Rosnes}.
 
 Codes $\mathcal{C}_{11}$ and $\mathcal{C}_{12}$ are two variable-regular codes constructed by PEG algorithm \cite{Peg} (available in \cite{mackayencyclopedia}). It takes $10$ and $35$ minutes to find upper bounds of $19$ and $37$ on the stopping distance of $\mathcal{C}_{11}$ and $\mathcal{C}_{12}$, respectively. 
 For the purpose of comparing the run-times, we note that, it took the algorithm of~\cite{Hirotomo2}, about $25$ hours to find the same upper bound for $\mathcal{C}_{11}$. This bound  also matches the exact stopping distance reported in \cite{Rosnes}. Also, to the best of our knowledge, the upper bound of $37$ on the $s_{min}$ of $\mathcal{C}_{12}$ has not been reported in the literature so far. 
 
 Moreover, the exact stopping distance of $\mathcal{C}_{14}$, a high-rate structured code with block length $5219$ has been found in just $913$ seconds.
 
In~\cite{tas}, the authors constructed QC-LDPC codes that are cyclic liftings of fully-connected $3 \times n$ protographs, and have the shortest block length for a given girth. Code $C_{15}$ is the shortest
cyclic lifting of the $3 \times 6$ fully-connected base graph with girth $10$, reported in~\cite{tas}. We find $s_{\min}$ of this code to be $14$ in $28$ seconds.
 Code $\mathcal{C}_{16}$ is the Ramanujan $(4896,2448)$ code with $g=12$. For this code, we find 
the exact value of $s_{\min}$ to be $24$, in about $12$ minutes. For the purpose of comparing the run-times, we note that, it took the algorithm of~\cite{Hirotomo2}, about $162$ hours to find the same upper bound for $\mathcal{C}_{16}$.

 Recently, QC-LDPC codes with girth $8$, whose parity-check matrices have some symmetries, and are in many cases shorter than previously existing girth-$8$ QC-LDPC codes, were constructed in~\cite{tas2}. 
  We tested the codes of~\cite{tas2}, and observed that our proposed algorithm can find the exact $s_{\min}$, or obtain lower and upper bounds on $s_{\min}$, for many of them in a matter of seconds or minutes.
 For example, we have found the stopping  distances of all $18$ codes with $d_\mathrm{v}=3$, $g=8$, $R \leq 0.88$ and $n \leq 4000$ (Table I of \cite{tas2}), each in less than or about one minute. The last code in that table is $\mathcal{C}_{13}$ in Table \ref{tab:dss3}.  
 
 While most of the variable-regular codes studied in the literature, see, e.g., \cite{Rosnes}, have $d_{\mathrm{v}}=3$, the algorithms proposed here can find the exact stopping distance, or provide lower and upper bounds on stopping distance of variable-regular codes with $d_\mathrm{v}>3$. 
As an example, we are able to provide lower bounds on, or obtain the exact value of, $s_{\min}$ for all the variable-regular LDPC codes provided in Tables II and III of \cite{tas2}, in just a few minutes. These are codes with variable degrees $4$ and $5$, respectively, and with  $R \leq 0.84$ and $n \leq 14750$. In many cases, also, we find upper bounds on $s_{\min}$ for these codes. Four examples of the codes in Tables II and III of \cite{tas2} are listed as the last entries of Table~\ref{tab:dss3}.

Based on the value of stopping distance, block length, rate and degree distribution of the reported codes in the literature \cite{Rosnes}, \cite{richter}, \cite{Hu2}, \cite{Hirotomo2}, we believe finding the exact (or bounds on the) stopping distance of codes such as $\mathcal{C}_3$, $\mathcal{C}_8$, $\mathcal{C}_9$, $\mathcal{C}_{12}$, $\mathcal{C}_{14}$, $\mathcal{C}_{18}$ and $\mathcal{C}_{20}$ are out of the reach of their algorithms or the run-times will be significantly larger than ours.

We have used Algorithm~\ref{algup} to provide an upper bound on the stopping distance of eight irregular codes listed in Table \ref{tab:dssireg}.
\begin{table}[]
\centering
\renewcommand{\arraystretch}{0.9}
\setlength{\tabcolsep}{3pt}
\caption{Upper Bounds on the Stopping Distance of some Irregular LDPC Codes}
\label{tab:dssireg}
\begin{tabular}{||c|c|c|c||c|c||}
\cline{1-6}
Code&Girth&Rate&Length&Upper Bound&$b'_{\max}$\\ 
\cline{1-6}
$\mathcal{C}_{21}$\cite{CCSD}&6&0.5&128&$\begin{array}{@{}c@{}} 11(n) \\ \hdashline 118~sec. \end{array}$&8\\
\cline{1-6}
$\mathcal{C}_{22}$\cite{802.11}&6&0.83&648&$\begin{array}{@{}c@{}} 7(n) \\ \hdashline 24~sec. \end{array}$&7\\
\cline{1-6}
$\mathcal{C}_{23}$\cite{802.16}&6&0.83&1824&$\begin{array}{@{}c@{}} 8(e) \\ \hdashline 203~sec. \end{array}$&7\\
\cline{1-6}
$\mathcal{C}_{24}$\cite{802.16}&6&0.75&2304&$\begin{array}{@{}c@{}} 12(e) \\ \hdashline 165~sec. \end{array}$&8\\
\cline{1-6}
$\mathcal{C}_{25}$\cite{802.16}&6&0.67&2304&$\begin{array}{@{}c@{}} 15(e) \\ \hdashline 184~sec. \end{array}$&7\\
\cline{1-6}
$\mathcal{C}_{26}$\cite{802.11}&6&0.75&1944&$\begin{array}{@{}c@{}} 12(e) \\ \hdashline 289~sec. \end{array}$&8\\
\cline{1-6}
$\mathcal{C}_{27}$\cite{Peg}&8&0.5&1008&$\begin{array}{@{}c@{}} 13(e) \\ \hdashline 21~min. \end{array}$&10\\
\cline{1-6}
$\mathcal{C}_{28}$\cite{Peg}&8&0.5&2048&$\begin{array}{@{}c@{}} 15(e) \\ \hdashline 20~min. \end{array}$&10\\
\cline{1-6}
\end{tabular}
\end{table}
 Codes $\mathcal{C}_{21}-\mathcal{C}_{26}$ have been adopted in standards, and Codes $\mathcal{C}_{27}$ and $\mathcal{C}_{28}$ are random codes constructed by the PEG algorithm. 
In \cite{Rosnes2}, the exact stopping distance of all the IEEE 802.16e LDPC codes \cite{802.16} was reported. Our upper bound search algorithm can also find the same stopping distance in each case, 
most of the time in just a few seconds (no run-time for obtaining these results was reported in \cite{Rosnes2}). 

In this paper, we propose an efficient search algorithm to provide an exhaustive/non-exhaustive list of NETSs.
The results obtained from this search algorithm along with the theoretical results in \cite{hashemilower} support the assertion that in the harmful classes of TSs, 
there is no NETS (otherwise, they could potentially be harmful). For example, Table \ref{tab:38} shows the list of TSs  of Tanner (155,64) code \cite{Tanner3} ($d_\mathrm{v}=3$, $g=8$) 
in the wide range of $a \leq 13$ and $b \leq 4$. The multiplicities of LETS, ETSL and NETS structures are also shown separately in the table.
One can compare the number of LETS, ETSL and NETS in this  code to see that in the classes believed to be most harmful (with relatively small $a$ and $b$ values), the only TSs are LETSs.
\begin{table}[]
\centering
\caption{Multiplicities of $(a,b)$ TSs consisting of LETSs, ETSLs and NETSs of $\mathcal{C}_{10}$ within the range of $a \leq 13$ and $b \leq 4$}
\renewcommand{\arraystretch}{1.1}
\label{tab:38}
\begin{tabular}{||c|c|c|c|c|c|| }
\cline{1-6}
&\multicolumn{5}{c||}{$\mathcal{C}_{10}$}\\
\cline{1-6}
$(a,b)$&Total&Total&Total&Total&Total\\
class&LETS&ETSL&NETS&TS&TS\cite{zhang}\\
\cline{1-6}
(4,4)&465&0&0&465&465\\
\cline{1-6}
(5,3)&155&0&0&155&155\\
\cline{1-6}
(6,4)&930&1860&0&2790&2790\\
\cline{1-6}
(7,3)&930&0&0&930&930\\
\cline{1-6}
(8,2)&465&0&0&465&465\\
\cline{1-6}
(8,4)&5115&9300&0&14415&14415\\
\cline{1-6}
(9,3)&1860&3720&0&5580&5580\\
\cline{1-6}
(10,2)&1395&0&0&1395&1395\\
\cline{1-6}
(10,4)&29295&48360&5580&83235&83235\\
\cline{1-6}
(11,3)&6200&9300&1860&17360&17360\\
\cline{1-6}
(12,2)&930&0&0&930&930\\
\cline{1-6}
(12,4)&180885&134850&47895&363630&\textbf{36280}\\
\cline{1-6}
(13,3)&34875&5580&2790&43245&43245\\
\cline{1-6}
\end{tabular}
\end{table}
 Based on the value $a_{max}= 11$ from Table  \ref{tab:lowernets}, the results of NETSs presented in Table~\ref{tab:38} for the classes with $a \leq 11$ and $b \leq 4$ are exhaustive. 
 For finding NETSs beyond this range, in Algorithm \ref{algnets}, the exhaustive list of ETSs within the range  $a \leq 12$ and $b \leq 5$ has been used.
 In \cite{zhang}, authors used the \textit{branch-\&-bound} approach to propose an exhaustive search algorithm for finding TSs.  However, similar to the other branch-\&-bound algorithms, 
 this approach is only applicable to codes with short block lengths. The multiplicity of TSs in different classes found by our algorithm for Tanner (155,64) code is matched with 
 the one reported in \cite{zhang}, except for the $(12,4)$ class. While our search algorithm has found $363630$ TSs in the $(12,4)$ class, only $36280$ TSs have been reported in \cite{zhang}.\footnote{We believe that the result reported in \cite{zhang} should be a typographical error.}  
The run-time of the algorithm of \cite{zhang} to find the exhaustive list of  TSs is not reported. However, while it took the algorithm of \cite{zhang} $59$  minutes\footnote{This is the only run-time reported in \cite{zhang}. The run-time is for a standard desktop computer with a 2.67-GHz processor.} to find the TSs of a PEG code in the range of $a \leq 5$ and $b \leq 5$, our $dpl$ search algorithm finds the same set of TSs in less than 20 seconds. 
 
 \section{Conclusion}
\label{sec:conclude}
 
In this paper, we proposed a  hierarchical graph-based expansion approach to characterize non-elementary
trapping sets (NETS) of low-density parity-check (LDPC) codes. The proposed characterization is based
on \textit{depth-one tree (dot)} expansion technique. Each NETS structure $\mathcal{S}$ is characterized as a sequence of embedded NETS structures that starts from an ETS, and grows in each step by using a $dot$ expansion, until it reaches $\mathcal{S}$. The  characterization allowed us to devise efficient search algorithms for finding all the instances of $(a,b)$ NETS structures with $a \leq a_{max}$ and $b \leq b_{max}$, in a guaranteed fashion. 
The exhaustive search of NETSs along with the theoretical results provided in \cite{hashemilower} support the assertion that in the harmful classes of TSs, there is no NETS (otherwise, they could
potentially be harmful). We also devised a low-complexity non-exhaustive search algorithm for finding NETSs within a much wider range compared to the range for the exhaustive search. 

Moreover, in this paper, we derived tight lower and upper bounds on the stopping distance $s_{min}$ of
LDPC codes. The bounds, which were established using a combination of analytical results and search techniques, are applicable to LDPC codes with a wide range of rates and block lengths. 
To derive the bounds, we partitioned the stopping sets into two categories of elementary and non-elementary.  We noted that elementary stopping sets (ESSs) and non-elementary stopping sets (NESSs) are 
subset of leafless ETSs (LETSs) and NETSs, respectively. Using exhaustive LETS and NETS search algorithms, we searched the stopping sets of size less than $L$.
If the search happened to find a stopping set, then the smallest size of such a stopping set was $s_{min}$. Otherwise, if the search failed, then a lower bound of $L \leq s_{min}$ was established on $s_{min}$. 
For the upper bound, the LETS and NETS search algorithms were modified to increase the range of search for stopping sets with larger size at the expense of losing the exhaustiveness of the search. 
The proposed technique was applied to a large number of LDPC codes, and lower and upper bounds on $s_{min}$, and in many cases the exact value of $s_{min}$, were obtained in a matter of seconds or minutes. 
Many of such codes are out of the reach of the existing search-based algorithms that often have practical constraints on the block length, rate or the degree distribution of the codes that they can handle.

\end{document}